\documentclass{article}

\usepackage[utf8]{inputenc}
\usepackage[english]{babel}
\usepackage[T1]{fontenc}
\usepackage{lmodern}

\usepackage{amsmath, amsthm, amssymb}

\usepackage[affil-it]{authblk}

\usepackage{enumitem}
\usepackage{bussproofs}
\usepackage[colorlinks=true]{hyperref}
\usepackage[plain]{fancyref}
\usepackage[norefs,nocites]{refcheck}

\usepackage{tikz}
\usetikzlibrary{matrix}

\newcommand*{\fancyrefexlabelprefix}{ex}

\newcommand*{\frefexname}{\text{example}}
\newcommand*{\Frefexname}{\text{Example}}

\frefformat{vario}{\fancyrefexlabelprefix}{\frefexname\fancyrefdefaultspacing#1#3}
\Frefformat{vario}{\fancyrefexlabelprefix}{\Frefexname\fancyrefdefaultspacing#1#3}
\frefformat{plain}{\fancyrefexlabelprefix}{\frefexname\fancyrefdefaultspacing#1}
\Frefformat{plain}{\fancyrefexlabelprefix}{\Frefexname\fancyrefdefaultspacing#1}

\newcommand*{\fancyrefdeflabelprefix}{def}

\newcommand*{\frefdefname}{\text{definition}}
\newcommand*{\Frefdefname}{\text{Definition}}

\frefformat{vario}{\fancyrefdeflabelprefix}{\frefdefname\fancyrefdefaultspacing#1#3} 
\Frefformat{vario}{\fancyrefdeflabelprefix}{\Frefdefname\fancyrefdefaultspacing#1#3}
\frefformat{plain}{\fancyrefdeflabelprefix}{\frefdefname\fancyrefdefaultspacing#1}
\Frefformat{plain}{\fancyrefdeflabelprefix}{\Frefdefname\fancyrefdefaultspacing#1}

\newcommand*{\fancyrefthmlabelprefix}{thm}

\newcommand*{\frefthmname}{\text{theorem}}
\newcommand*{\Frefthmname}{\text{Theorem}}

\frefformat{vario}{\fancyrefthmlabelprefix}{\frefthmname\fancyrefdefaultspacing#1#3}
\Frefformat{vario}{\fancyrefthmlabelprefix}{\Frefthmname\fancyrefdefaultspacing#1#3}
\frefformat{plain}{\fancyrefthmlabelprefix}{\frefthmname\fancyrefdefaultspacing#1}
\Frefformat{plain}{\fancyrefthmlabelprefix}{\Frefthmname\fancyrefdefaultspacing#1}

\newcommand*{\fancyrefremlabelprefix}{rem}

\newcommand*{\frefremname}{\text{remark}}
\newcommand*{\Frefremname}{\text{Remark}}

\frefformat{vario}{\fancyrefremlabelprefix}{\frefremname\fancyrefdefaultspacing#1#3}
\Frefformat{vario}{\fancyrefremlabelprefix}{\Frefremname\fancyrefdefaultspacing#1#3}
\frefformat{plain}{\fancyrefremlabelprefix}{\frefremname\fancyrefdefaultspacing#1}
\Frefformat{plain}{\fancyrefremlabelprefix}{\Frefremname\fancyrefdefaultspacing#1}

\newcommand*{\fancyreflemlabelprefix}{lem}

\newcommand*{\freflemname}{\text{lemma}}
\newcommand*{\Freflemname}{\text{Lemma}}

\frefformat{vario}{\fancyreflemlabelprefix}{\freflemname\fancyrefdefaultspacing#1#3}
\Frefformat{vario}{\fancyreflemlabelprefix}{\Freflemname\fancyrefdefaultspacing#1#3}
\frefformat{plain}{\fancyreflemlabelprefix}{\freflemname\fancyrefdefaultspacing#1}
\Frefformat{plain}{\fancyreflemlabelprefix}{\Freflemname\fancyrefdefaultspacing#1}

\newcommand*{\fancyrefsubseclabelprefix}{subsec}

\newcommand*{\frefsubsecname}{\text{subsection}}
\newcommand*{\Frefsubsecname}{\text{Subsection}}

\frefformat{vario}{\fancyrefdeflabelprefix}{\frefsubsecname\fancyrefdefaultspacing#1#3}
\Frefformat{vario}{\fancyrefsubseclabelprefix}{\Frefsubsecname\fancyrefdefaultspacing#1#3}
\frefformat{plain}{\fancyrefsubseclabelprefix}{\frefsubsecname\fancyrefdefaultspacing#1}
\Frefformat{plain}{\fancyrefsubseclabelprefix}{\Frefsubsecname\fancyrefdefaultspacing#1}

\newcommand*{\fancyrefcorlabelprefix}{cor}

\newcommand*{\frefcorname}{\text{corollary}}
\newcommand*{\Frefcorname}{\text{Corollary}}

\frefformat{vario}{\fancyrefcorlabelprefix}{\frefcorname\fancyrefdefaultspacing#1#3}
\Frefformat{vario}{\fancyrefcorlabelprefix}{\Frefcorname\fancyrefdefaultspacing#1#3}
\frefformat{plain}{\fancyrefcorlabelprefix}{\frefcorname\fancyrefdefaultspacing#1}
\Frefformat{plain}{\fancyrefcorlabelprefix}{\Frefcorname\fancyrefdefaultspacing#1}

\newcommand*{\fancyrefsubsubseclabelprefix}{subsubsec}

\newcommand*{\frefsubsubsecname}{\text{subsubsection}}
\newcommand*{\Frefsubsubsecname}{\text{Subsubsection}}

\frefformat{vario}{\fancyrefdeflabelprefix}{\frefsubsubsecname\fancyrefdefaultspacing#1#3}
\Frefformat{vario}{\fancyrefsubsubseclabelprefix}{\Frefsubsubsecname\fancyrefdefaultspacing#1#3}
\frefformat{plain}{\fancyrefsubsubseclabelprefix}{\frefsubsubsecname\fancyrefdefaultspacing#1}
\Frefformat{plain}{\fancyrefsubsubseclabelprefix}{\Frefsubsubsecname\fancyrefdefaultspacing#1}

\newcommand*{\fancyrefalgolabelprefix}{algo}

\newcommand*{\frefalgoname}{\text{algorithm}}
\newcommand*{\Frefalgoname}{\text{Algorithm}}

\frefformat{vario}{\fancyrefalgolabelprefix}{\frefalgoname\fancyrefdefaultspacing#1#3}
\Frefformat{vario}{\fancyrefalgolabelprefix}{\Frefalgoname\fancyrefdefaultspacing#1#3}
\frefformat{plain}{\fancyrefalgolabelprefix}{\frefalgoname\fancyrefdefaultspacing#1}
\Frefformat{plain}{\fancyrefalgolabelprefix}{\Frefalgoname\fancyrefdefaultspacing#1}

\newcommand*{\fancyrefproplabelprefix}{prop}

\newcommand*{\frefpropname}{\text{proposition}}
\newcommand*{\Frefpropname}{\text{Proposition}}

\frefformat{vario}{\fancyrefproplabelprefix}{\frefpropname\fancyrefdefaultspacing#1#3}
\Frefformat{vario}{\fancyrefproplabelprefix}{\Frefpropname\fancyrefdefaultspacing#1#3}
\frefformat{plain}{\fancyrefproplabelprefix}{\frefpropname\fancyrefdefaultspacing#1}
\Frefformat{plain}{\fancyrefproplabelprefix}{\Frefpropname\fancyrefdefaultspacing#1}

\newtheorem{counter}{Counter}

\theoremstyle{definition}

\newtheorem{definition}[counter]{Definition}

\theoremstyle{plain}
\newtheorem{lemma}[counter]{Lemma}
\newtheorem{theorem}[counter]{Theorem}
\newtheorem{corollary}[counter]{Corollary}

\theoremstyle{remark}
\newtheorem{example}[counter]{Example}
\newtheorem{remark}[counter]{Remark}

\let\oldtextbf\textbf
\renewcommand{\textbf}[1]{\oldtextbf{\boldmath #1}}

\newcommand\myLL[1]{\LeftLabel{$^{\text{#1}}$}}

\newcommand\tabProb{\textbf{prob}}
\newcommand\notT{\textbf{notT}}
\newcommand\notF{\textbf{notF}}
\newcommand\andT{\textbf{andT}}
\newcommand\andF{\textbf{andF}}

\newcommand\basis{\mathcal{B}}
\newcommand\Prop{\mathsf{Prop}}
\newcommand\Rat{\mathbb{Q}}
\newcommand\Nat{\mathbb{N}}
\newcommand\Tm{\mathsf{Tm}}
\newcommand\Con{\mathsf{Con}}

\newcommand\lanPPJ{\mathcal{L_{\PPJ}}}
\newcommand\lan{\mathcal{L}}
\newcommand\lanPL{\mathcal{L}_{\PL}}
\newcommand\lanPPL{\mathcal{L}_{\PPL}}

\newcommand\lanMod{\mathcal{L}_{\Box}}
\newcommand\lanJ{\mathcal{L}_{\J}}

\newcommand\axP{\mathsf{(P)}}
\newcommand\axJ{\mathsf{(J)}}
\newcommand\axPlus{\mathsf{(+)}}
\newcommand\ANE{\mathsf{(AN!)}}
\newcommand\pos{\mathsf{(NN)}}
\newcommand\LOne{\mathsf{(L1)}}
\newcommand\LTwo{\mathsf{(L2)}}
\newcommand\AddOne{\mathsf{(Add1)}}
\newcommand\AddTwo{\mathsf{(Add2)}}

\newcommand\MP{\mathsf{(MP)}}
\newcommand\CE{\mathsf{(PN)}}
\newcommand\ST{\mathsf{(ST)}}

\newcommand\logic{\mathsf{L}}
\newcommand\CP{\mathsf{CP}}
\newcommand\D{\mathsf{D}}
\newcommand\PPJ{\mathsf{PPJ}}
\newcommand\PJ{\mathsf{PJ}}
\newcommand\PCP{\mathsf{PCP}}
\newcommand\PPCP{\mathsf{PPCP}}
\newcommand\PL{\mathsf{PL}}
\newcommand\PPL{\mathsf{PPL}}
\newcommand\J{\mathsf{J}}
\newcommand\K{\mathsf{K}}
\newcommand\LPPS{\mathsf{LPP^S}}
\newcommand\LPPTwo{\mathsf{LPP_2}}
\newcommand\LPPOne{\mathsf{LPP_1}}
\newcommand\LPlogic{\mathsf{LP}}

\newcommand\system{\mathcal{S}}

\newcommand\p{\mathsf{P}}
\newcommand\pspace{\mathsf{PSPACE}}

\newcommand\np{\mathsf{NP}}
\newcommand\conp{\mathsf{coNP}}
\newcommand\sig[2]{\mathsf{\Sigma^{#1}_{#2}}}

\newcommand\cclass{\mathsf{C}}

\newcommand\ih{\text{i.h.}}

\newcommand\myvec[1]{\boldsymbol{#1}}
\newcommand\op{\mathsf{\odot}}
\newcommand\algo{\mathsf{\mathcal{A}}}

\newcommand\mkr{\mathsf{M}}

\newcommand\true{\mathsf{T}}
\newcommand\false{\mathsf{F}}

\newcommand\powerset{\mathcal{P}}
\newcommand\subf{\mathsf{subf}}
\newcommand\evid{\mathcal{E}}
\newcommand\ext[1]{\bar{#1}}

\newcommand\sat{\mathsf{SAT}}
\newcommand\psat{\mathsf{PSAT}}
\newcommand\DSat{\D_{\sat}}
\newcommand\logicSat{\mathsf{L}_{\mathsf{SAT}}}
\newcommand\LPPTwoSat{\mathsf{LPP_{2,SAT}}}
\newcommand\LPPOneSat{\mathsf{LPP_{1,SAT}}}
\newcommand\JSat{\J_{\sat}}
\newcommand\PPLSat{\PPL_{\sat}}
\newcommand\PLSat{\PL_{\sat}}
\newcommand\PJSat{\PJ_{\sat}}
\newcommand\PPJSat{\PPJ_{\sat}}

\newcommand\cpnb{\mathsf{cpnb}}
\newcommand\cpnf{\mathsf{cpnf}}

\title{The Complexity of Satisfiability in 
Non-Iterated and Iterated Probabilistic Logics}

\author{Ioannis Kokkinis}

\affil{Technical University of Dortmund\\
\href{mailto:ioannis.kokkinis@tu-dortmund.de}{ioannis.kokkinis@tu-dortmund.de} 
}

\begin{document}

\maketitle

\begin{abstract}
Let L be some extension of
classical propositional logic. The non-iterated
probabilistic logic over L is the logic PL
that is defined by adding non-nested
probabilistic operators in the language of L. For example, in PL we can
express a statement like ``the probability of
truthfulness of A is at least 0.3'' where A is
a formula of L.
The iterated probabilistic logic over
L is the logic PPL, where the probabilistic
operators may be iterated (nested). For example,
in PPL we can express a statement like
``this coin is counterfeit with probability 0.6''.
In this paper we investigate the influence of
probabilistic operators in the complexity of
satisfiability in PL and PPL. 
We obtain complexity bounds, for the aforementioned
satisfiability problem, which
are parameterized in the complexity of satisfiability of conjunctions of positive
and negative formulas
that have neither a probabilistic nor a classical
operator as a top-connective.
As an application of our results we obtain tight complexity bounds
for the satisfiability problem in PL and PPL when
L is classical propositional logic or justification logic.

\end{abstract}

\textbf{keywords:} probabilistic logic, computational complexity, satisfiability, justification logic

\textbf{Acknowledgements.} The author is grateful to 
Antonis Achilleos and Thomas Studer for several
useful discussions and suggestions and to the anonymous
reviewers for many valuable comments and remarks that
helped him improve the quality of this
paper substantially. The initial part of 
this research was conducted when the author was
at the Institute of Computer Science (INF) in
Bern, Switzerland (financial support by the
Swiss National Science Foundation
project 153169, \emph{Structural Proof Theory and the Logic of Proofs}), and the final part when the author
was at the Lorraine Research Laboratory in 
Computer Science and its Applications (LORIA) in 
Nancy, France (financial support by the European 
Research Council project 313360, \emph{Epistemic 
Protocol Synthesis}).

\section{Introduction}

\subsection{Background and Related Work}

Probabilistic logics (also known as probability logics) are logics that can be
used to model uncertain reasoning. Although the
idea of probabilistic logic was first proposed by Leibniz~\cite{leibniz89},
the modern development of this topic started only in the 1970s 
and 1980s in the papers of
H. Jerome Keisler~\cite{kei77} and Nils Nilsson~\cite{nil86}. 
Following Nilsson's research, Fagin, Halpern and
Meggido~\cite{fahame90} introduced a logic with arithmetical 
operations built into the syntax
so that Boolean combinations of linear inequalities of 
probabilities of formulas
can be expressed. Based on Nilsson's research
Ognjanovi\'{c}, Ra\v{s}kovi\'{c} and 
Markovi\'{c}~\cite{ograma09} defined the logic
$\LPPTwo$, which is a non-iterated probabilistic logic with
classical base. The language of
$\LPPTwo$ is defined by adding (non-nested) 
operators of the form $P_{\geq s}$ (where $s$ is
a rational number) to the language of
classical propositional logic. In $\LPPTwo$ we can
have expressions of the form $P_{\geq s} \alpha$, which
read as ``the probability of truthfulness
of classical propositional formula
$\alpha$ is at least $s$''. In addition
to $\LPPTwo$, the authors of~\cite{ograma09} define the logic
$\LPPOne$, which is a probabilistic logic over
classical propositional logic, that allows
iterations (nesting) of the probabilistic operators ($P_{\geq s}$).
In $\LPPOne$ we can describe a situation like the following: let $c$ be a coin and let $p$
be the event ``$c$ lands tails''.
Assume that the probability 
of $c$ landing tails is at least
$60\%$ (because $c$ is counterfeit). We can express
this fact in non-iterated probabilistic
logics with the formula $P_{\geq 0.6} p$.
Assume now that we are uncertain about the
fact of $c$ being counterfeit.
In order to express this statement we need
nested applications of the probabilistic
operators. In $\LPPOne$ for example we can have a formula
like $P_{\geq 0.8} P_{\geq 0.6} p$. 

In addition to classical propositional logic,
probabilistic logics have been defined over several other logics
(see the recent~\cite{ognjanovicRM16} for an overview).
For example in \cite{komaogst,koogst}  we defined two
probabilistic logics over justification logic (\emph{probabilistic
justification logics} for short). Justification logic~\cite{artemovF16} can be understood as an explicit analogue of
modal logic~\cite{bmv01}. 
Whereas traditional modal logic uses formulas
of the form $\Box \alpha$ to express that an agent believes 
$\alpha$, the language of justification logic
`unfolds' the $\Box$-modality into a family of so-called 
\emph{justification terms}, which
are used to represent evidence for the agent's belief.
Hence, instead of $\Box \alpha$,
justification logic includes formulas of the form $t : \alpha$, where $t$ is a justification
term.  Formulas of the form $t: \alpha$ are
called \emph{justification assertions} and read as
\begin{center}
the agent believes $\alpha$ for reason $t$.
\end{center}
Justification terms can represent any kind of objects that we use as evidence:
for example proofs in Peano arithmetic or informal justifications (like everyday observations,
texts in newspapers, or someone's words).
Artemov developed the first 
justification  logic, the Logic of Proofs (usually 
abbreviated as $\LPlogic$), 
to provide intuitionistic logic with a 
classical provability 
semantics~\cite{Art95TR,Art01BSL}. Except from $\LPlogic$,
several other justification logics have been introduced.
The minimal justification logic is called $\J$~\cite{artemovF16}.
By the famous realization theorem~\cite{Art01BSL,Bre00TR}
$\J$ corresponds to the minimal modal logic
$\K$. That is, we can translate any theorem of $\J$ to a theorem
of $\K$ by replacing any term with the $\Box$ and also any theorem
of $\K$ to a theorem of $\J$ by replacing any occurrence
of $\Box$ with an appropriate justification term.

The non-iterated probabilistic logic over $\J$, the logic
$\PJ$, is defined in \cite{komaogst} and the iterated
probabilistic logic over $\J$, the logic $\PPJ$, is defined in \cite{koogst}. In $\PJ$ we can describe
a situation like the following:
assume that an agent reads in some reliable
newspaper that fact $\alpha$ holds and also that the agent hears
that fact $\alpha$ holds from some
unreliable neighbour. Then, the
agent has two justifications for $\alpha$: the text of
the newspaper, represented by $s$, and the words of 
their neighbour, represented by $t$.
We can express the fact that the newspaper is a more
reliable source than the neighbour using the $\PJ$-formulas
$P_{\geq 0.8} (s : \alpha)$ and $P_{\geq 0.2} (t : \alpha)$.
So, we can use probabilistic justification logic to model the idea that
\begin{center}
different kinds of evidence for $\alpha$
lead to different degrees of belief in $\alpha$.
\end{center}
It is tempting to try to model the above idea using formulas
of the form $s \to P_{\geq 0.8} \alpha$ (which are not allowed in the syntax of $\PJ$).
However this approach treats justifications as statements: in order for $P_{\geq 0.8} \alpha$
to hold, $s$ has to hold too. In our approach this is not necessary. We believe that
justifications do not need to be true, simply because one might want to
believe something for a false reason. Also this approach
places the uncertainty on top of formula $\alpha$. The approach of $\PJ$ places the
uncertainty where it should be: on the the fact that $s$ justifies $\alpha$. 

The most interesting property of the logic $\PPJ$ is the fact that the language
of $\PPJ$ allows applying
justification terms to probabilistic operators and vice versa (as we will see
later this is the property that makes finding
complexity bounds for the satisfiability problem in $\PPJ$ a challenging task).
So, continuing our example with the counterfeit coin, if $p$ is the event
``the coin lands tails'', and $t$ is some explicit reason to believe that,
then in $\PPJ$ we could have the formula $P_{\geq 0.8} (t : P_{\geq 0.3}p)$,
with a meaning like ``I am uncertain for a particular justification of this coin
being counterfeit, e.g. because this coin looks similar to a counterfeit coin
I have seen some time ago''.
As another application of $\PPJ$, in \cite{koogst} we have shown that the
lottery paradox~\cite{kyburg61} can be analysed in this logic.
The lottery paradox goes as follows: assume that
we have $1,000$ tickets in a lottery where every
ticket has the same probability to win and there is
exactly one winning ticket. Now assume a proposition is believed if and only 
if its degree of belief is greater than 0.99. In this setting it 
is rational to believe that ticket~1 does not win, it is 
rational to believe that ticket~2 does not win, and so on. 
However, this entails that it is rational to believe that no 
ticket wins because rational belief is closed under conjunction.
Hence, it is rational to believe that no ticket wins and (of course) it is rational
to believe that one  ticket wins, which is absurd. In \cite{koogst} we have
formalized the lottery paradox in $\PPJ$ and we have also proposed
a solution for avoiding the paradox via
restricting the axioms that are justified in $\PPJ$.

A model for a non-iterated probabilistic logic is a probability space
where the events are models of the base logic.
A model for an iterated probabilistic logic is a probability space where
the events contain models of the base logic and other probability spaces, so
that we can deal with iterated probabilities. One can say that the models for iterated probabilistic
logics look like Kripke structures, where the accessibility relation is replaced by
a probability measure. 
The satisfiability problem
for a probabilistic logic is to decide whether such a model that satisfies a given formula
exists. In the 1980s
Georgakopoulos et al.~\cite{georgakopoulosKP88}
studied a problem that is very similar to
the satisfiability problem in probabilistic logics.
This problem is called $\psat$ and it is
a probabilistic version of the famous satisfiability
problem in classical propositional logic (i.e. the well known  $\sat$-problem~\cite{papad94}).
The problem $\psat$ can be formalized as follows:
assume that we are given a formula in conjunctive normal form and
a probability for each clause. Is there a probability distribution (over the set of all 
possible truth assignments of the variables appearing in the clauses) that
satisfies all the clauses? Georgakopoulos et al.  reduced $\psat$
to solving a linear system, and proved that
$\psat$ is $\np$-complete. 
Although the expressive power of the
formal systems of~\cite{fahame90} and
\cite{ograma09} is richer than the one 
of~\cite{georgakopoulosKP88},
the authors of~\cite{fahame90} and
\cite{ograma09} were able to use 
arguments similar to those in \cite{georgakopoulosKP88}
to show that
the satisfiability problem in their logics is 
also $\np$-complete.
In \cite{kokkinis16} we obtained tight bounds
for the complexity of the satisfiability problem
in non-iterated
probabilistic justification logic, using again
some results from the theory of linear programming.
Fagin and Halpern~\cite{faginH94} mention (without giving a complete
formal proof) that complexity
bounds for the satisfiability problem in a modal logic
that allows nesting of the probabilistic operators (like in $\LPPOne$) can be obtained by
employing an algorithm based on a tableau construction
as in classical modal logic~\cite{halmos92}. In \cite{kokkinis17}
we used the idea of Fagin and Halpern in order to
obtain tight bounds for the complexity of satisfiability
in $\PPJ$.

\subsection{Our Contribution}

The goal of this paper is to summarize
and generalize the results
for the complexity of the satisfiability
problem in non-iterated and iterated probabilistic logics.
This paper is the extended journal version of
\cite{kokkinis16} and \cite{kokkinis17}
which were presented at ``Foundations of Information and Knowledge Systems" in 2016 and the ``11th Panhellenic Logic Symposium" in 2017 respectively. The results of
\cite{kokkinis16,kokkinis17} refer only to
probabilistic justification logic, whereas in the present
paper we make clear that our results can be applied
to an non-iterated and iterated probabilistic
logic over any extension of classical propositional logic.
Whereas the result of \cite{kokkinis16} is a straightforward
adaptation of some arguments from \cite{fahame90},
the result of \cite{kokkinis17} is new and non-trivial.
The fact that a tableaux method can be used for obtaining complexity bounds
for iterated probabilistic logics was already observed in \cite{faginH94}, but no formal
proof was given. In the short conference papers \cite{koogst} and \cite{kokkinis17}
we gave decidability and complexity proofs respectively for iterated probabilistic
justification logic. In the present paper we give the complexity (and thus
decidability proof) for $\PPJ$ in full detail using a tableaux method. More precisely, we present
upper and lower complexity bounds, for the aforementioned satisfiability 
problem, which are parameterized on the complexity of
satisfiability of conjunctions of positive
and negative formulas that have neither a probabilistic
nor a classical operator as a top-connective. We also show how
our results can be applied to the special cases where the
probabilistic logics are defined over classical propositional logic
or justification logic.

\subsection{Outline of the Paper}

In \Fref{sec:pre} we give some preliminary
definitions and prove a lemma from the theory of
linear programming that is necessary for our analysis.
In Sections~\ref{sec:non-iter} and
\ref{sec:iter} we obtain complexity bounds
for the satisfiability problem in non-iterated
and iterated probabilistic logics over
any extension of classical propositional
logic. In \Fref{sec:appl} we apply
the results of Sections~\ref{sec:non-iter} and
\ref{sec:iter} to determine the complexity of
satisfiability in
probabilistic logics over classical propositional
logic and over justification logic. In 
\Fref{sec:concl} we present our final remarks and
present some directions for further research.

\section{Preliminaries}

\label{sec:pre}

For the purposes of this paper a \emph{logic} is a
formal system, defined via a set of axioms and inference 
rules, a notion of semantics (i.e. a formal
definition of the notion of \emph{model} for the
logic), together with a provability and satisfiability
relation over some formal language. In this paper we
are interested in obtaining complexity bounds for
the following decision problem:

\begin{definition}[Satisfiability Problem]
Let $\logic$ be a logic over some language $\lan$. 
The satisfiability problem
for $\logic$ (denoted as $\logicSat$) is 
the following problem:
\begin{center}
given some $\alpha \in \lan$, is there a model
of $\logic$ that satisfies $\alpha$?
\end{center}
\end{definition}

For a formula $\alpha$ in the language of some logic $\logic$,
$\alpha$ is \emph{satisfiable} means that there is an 
$\logic$-model that satisfies $\alpha$. If the satisfiability in $\logic$ is defined in worlds of
the models, then $\alpha$ is satisfiable means that there is
a model of $\logic$, $M$, and a world $w$, such that $\alpha$ is 
satisfied in the world $w$ of $M$.
Since the satisfiability problem depends only
on semantical notions, we will present all the logics without the corresponding axiomatization.
The only exceptions are the
basic justification logic $\J$ and the iterated
probabilistic logic over $\J$, $\PPJ$, where
it is necessary to know what the axioms
of the logic are, for properly defining the models. 

All the logics in this paper are extensions of classical propositional logic. 
The following definition is very important for our analysis.

\begin{definition}[Basic Formulas]
Let $\logic$ be a logic over language $\lan$. The basic formulas of 
$\lan$ (represented as $\basis(\lan)$) are the formulas of $\lan$ that do not have
$\lnot$, $\land$ or a probabilistic operator $P_{\geq s}$ (the probabilistic
operators will be formally defined later) as their top-connectives. 
We assume  that $\basis(\lan)$ contains at least
$\Prop$, which is a countable set of atomic
propositions. We will
refer to the elements of $\basis(\lan)$ as the 
basic formulas of language $\lan$ or the basic formulas
of logic $\logic$.
\end{definition}

In the rest of the paper we fix a logic $\logic$ over
a language $\lan$. We assume that
$\logic$ is an extension of classical propositional
logic and that $\lan$ is defined by 
the following grammar:
\[
\alpha :: = b ~|~ \lnot \alpha ~|~ \alpha \land \alpha,
\]
where $b \in \basis(\lan)$. We assume that we are given
a function $v$ which assigns a truth value ($\true$ for
true and $\false$ for false) to elements of
 $\basis(\lan)$.
The extension of $v$
to the elements of $\lan$ is the function
$\ext{v}$, which is defined classically. Sometimes, we will abuse
notation and use the symbol $v$ in place of $\ext{v}$.
We will refer to $v$ (or its extension) as an \emph{evaluation}.
We use Greek lower-case letters like
$\alpha$, $\beta$, $\gamma$, $\ldots$ for
members of $\lan$. The symbol $\powerset$
stands for powerset.
We also define the following abbreviations in the
standard way:
\begin{align*}
\alpha \lor \beta &\equiv \lnot (\lnot \alpha \land
\lnot \beta)~; \\
\alpha \to \beta &\equiv \lnot \alpha \lor \beta~.
\end{align*}

From the above discussion it is clear that in order to define the
semantics of $\logic$ it suffices to determine which are the basic formulas
and how the evaluation behaves on them. For example, if we assume that the
basic formulas are atomic propositions (i.e. elements of $\Prop$) and that
the evaluation is a classical truth assignment, then we have defined the language and
semantics of classical propositional logic.

In the next sections we define probabilistic
logics over $\logic$. Models for these logics are
probability spaces where the events are models for $\logic$ (and in the
iterated case contain other probability spaces too).
In order to formally present these models, we need the following definitions:

\begin{definition}[$\sigma$-Algebra Over a Set]
Let $W$ be a non-empty set and let $H$ be a non-empty 
subset of $\powerset(W)$. We call
$H$ a \emph{$\sigma$-algebra over $W$}
if the following hold:
\begin{itemize}[topsep = 0em]
\item
$W \in H$~;
\item
$U \in H \Longrightarrow W \setminus U \in H$~.
\item
For any countable collection of elements of $H$, $U_0, U_1, \ldots$,
it holds that:
\[
\bigcup_{i \in \mathbb{N}} U_i \in H~.
\]
\end{itemize}
\end{definition}

\begin{definition}[$\sigma$-Additive Measure]
Let $H$ be a $\sigma$-algebra over $W$ and assume that $\mu : H \to [0,1]$.
We call $\mu$  a \emph{$\sigma$-additive measure} if the following hold:
\begin{enumerate}[topsep = 0em, label = (\arabic*)]
\item
$\mu(W) = 1$.
\item
Let $U_0, U_1, \ldots$ be a countable collection of pairwise disjoint elements of $H$. Then:
\[
\mu \left (\bigcup_{i \in \mathbb{N}} U_i \right ) = \sum_{i \in \mathbb{N}} \mu(U_i).
\]
\end{enumerate}
\end{definition}

\begin{definition}[Probability Space]
A \emph{probability space} is a structure
$\langle W, H, \allowbreak \mu\rangle$,
where:
\begin{itemize}[topsep=0em]
\item
$W$ is a non-empty set;
\item
$H$ is a $\sigma$-algebra over $W$;
\item
$\mu : H \to [0,1]$ is a $\sigma$-additive measure.
\end{itemize}
The members of $H$ are called \emph{measurable} sets.
\end{definition}

A \emph{finitely additive} measure can be defined by assuming a finite, instead of
a countable, union in the previous definitions.
Semantics for probabilistic logics over classical
propositional logic has been given both for
$\sigma$- and for finitely additive measures~\cite{ograma09}. Semantics
for probabilistic logics over justification logic~\cite{komaogst,koogst}
has been given only for finitely additive measures. However,
after the small model theorems that we will prove,
the probability spaces in the models will be finite, so the results of this paper
hold for the finitely additive case too.

As we mentioned in the introduction, decidability and
complexity results in probabilistic logics
heavily depend on results from the theory
of linear programming.
In this paper we will use a theorem
that provides bounds on
the size of a solution of a linear system
using the
sizes of the constants that appear in the system.
Before showing this result, we need to define the size for non-negative
integers and rational numbers and to present
\Fref{thm:lin_eq_thm}.
We use \textbf{bold} font for vectors.
The superscript $*$ in a vector denotes
that the vector represents a solution of some
linear system. 

\begin{definition}[Sizes]
Let $r$ be a non-negative integer.
The size of $r$, represented as $|r|$,
is the number of bits needed for representing
$r$ in the binary system. If 
$r = \frac{s_1}{s_2}$ is a
rational number, where $s_1$ and
$s_2$ are relatively prime non-negative
integers with $s_2 \neq 0$, then the size
of $r$ is
$|r| :=  |s_1| + |s_2|$.
\end{definition}

\begin{theorem}[\text{\cite[p. 145]{chvatal83}}]
\label{thm:lin_eq_thm}
Let $\system$ be a system of $r$ linear equalities. Assume 
that the vector $\myvec{x^*}$
is a solution of $\system$ such that all of
$\myvec{x^*}$'s entries are non-negative.
Then there is a vector $\myvec{y^*}$ such that
\begin{enumerate}[label=(\arabic*), topsep = 0em]
\item
$\myvec{y^*}$ is a solution of $\system$;
\item
all the entries of $\myvec{y^*}$ are non-negative;
\item
at most $r$ entries of $\myvec{y^*}$ are positive.
\end{enumerate}
\end{theorem}

\Fref{thm:lin_ineq_eq_thm} provides the anounced bounds
on the solution of a linear system.
A sketch of its proof was given
in \cite[Lemmata 2.5 and 2.7]{fahame90}.
To make our presentation complete, we provide
a detailed proof here.

\begin{theorem}
\label{thm:lin_ineq_eq_thm}
Let $\system$ be a linear system of
$n$ variables and of
$r$ linear equalities and/or 
inequalities with integer coefficients
each of size at most $l$.
Assume that the vector \mbox{$\myvec{x^*} =
x^*_1, \ldots, x^*_n$}
is a solution of $\system$ such that for all 
$i \in \{ 1, \ldots, n\}$, $x^*_i \geq 0$.
Then, there is a vector $\myvec{y^*} = y^*_1, \ldots, y^*_n$
that satisfies the following properties
\begin{enumerate}[label=(\arabic*), topsep = 0.3em]
\item
$\myvec{y^*}$ is a solution of $\system$;
\item
at most $r$ entries of $\myvec{y^*}$ are positive;
\item
for all $i$, $y^*_i$ is a non-negative
rational number with size bounded 
by
\[
2 \cdot \big ( r \cdot l+ r \cdot \log_2 (r) + 1 
\big )~.
\]
\end{enumerate}
\end{theorem}

\begin{proof}
We make the following conventions:

\begin{itemize}
\item 
All vectors used in this proof have $n$ entries. 
The entries of the vectors
are assumed to be in one to one
correspondence with the variables that appear in the
original system $\system$.
\item 
Let $\myvec{y}^*$ be
a solution of a linear system $\mathcal{T}$.
If $\myvec{y}^*$ has more entries than the 
variables of $\mathcal{T}$
we imply that
entries of $\myvec{y}^*$ that correspond to
variables appearing in $\mathcal{T}$
compose a solution of $\mathcal{T}$.
\item 
Assume that system
$\mathcal{T}$ has less variables than system
$\mathcal{T}'$. When we say that any solution
of $\mathcal{T}$ is a solution of $\mathcal{T}'$ we imply
that the missing variables are set to $0$.
\end{itemize}
Assume that the original
system $\system$ contains an inequality of
the form:
\begin{equation}
\label{eq:ineqEx}
b_1 \cdot x_{1} + \ldots + b_n \cdot x_{n} 
~\op~ c~,
\end{equation}
for $\op \in \{ < , \leq , \geq, >\}$ where
$x_{1}, \ldots , x_n$ are variables 
and $b_1 , \ldots ,$
$b_n, c$ are constants that appear in
$\system$. Vector
$\myvec{x^*}$ is a solution of \eqref{eq:ineqEx}.
We replace the inequality
\eqref{eq:ineqEx} in $\system$ with the
following equality:
\[
b_1 \cdot x_1 + \ldots + b_n \cdot x_{n} = 
b_1 \cdot x^*_1 + \ldots + b_n \cdot x^*_n~.
\]
We repeat this procedure for every inequality of $\system$.
This way we obtain a system of linear equalities which we call $\system_0$. 
It is easy to see that $\myvec{x^*}$ is
a solution of $\system_0$ and
that any solution of $\system_0$ is also a
solution of $\system$.

Now we will transform $\system_0$ to another 
linear system by applying the following
algorithm:
\begin{enumerate}[label = (\roman*)]
\item
Set $i := 0$, $e_0 := r$, $v_0 := n$, 
$\myvec{x^{*,0}} := \myvec{x^*}$.
Go to step \ref{enum:linTransEqVarCheck}.
\item
\label{enum:linTransEqVarCheck}
If $e_i = v_i$ then go to step
\ref{enum:linTransEqNonZeroDetVarCheck}.
Otherwise go to step \ref{enum:linTransIneqVarCheck}.
\item
\label{enum:linTransEqNonZeroDetVarCheck}
If the determinant
of $\system_i$ is non-zero then stop.
Otherwise go to step \ref{enum:linTransDepend}.
\item
\label{enum:linTransIneqVarCheck}
If $e_i < v_i$ then go to step \ref{enum:linTransEqThm},
else go to step \ref{enum:linTransDepend}.
\item
\label{enum:linTransEqThm}
We know that the vector $\myvec{x^{*,i}}$
is a non-negative
solution for the system $\system_i$.
From \Fref{thm:lin_eq_thm} we obtain a solution $\myvec{x^{*,i+1}}$
for the system
$\system_i$ which has at most $e_i$ entries positive.
In $\system_i$ we replace the variables that
correspond to zero entries of the solution $\myvec{x^{*,i+1}}$
with zeros. We obtain a new system
which we call $\system_{i+1}$ with 
$e_{i+1} = e_i$ equalities
and $v_{i+1} = e_i < v_i$ variables. 
Vector $\myvec{x^{i+1}}$ is a solution of $\system_{i+1}$
and any solution of $\system_{i+1}$ is a solution
of $\system_i$.
We set $i := i+1$ and
we go to step \ref{enum:linTransEqVarCheck}.
\item
\label{enum:linTransDepend}
We remove only one equation that
can be written as a linear combination of some
others. We obtain a new system
which we call $\system_{i+1}$ with 
$e_{i+1} = e_i -1$ equalities
and $v_{i+1} = v_i$ variables. We set $i := i+1$ and
$\myvec{x^{*,i+1}} := \myvec{x^{*,i}}$.
We go to step \ref{enum:linTransEqVarCheck}.
\end{enumerate}
From steps \ref{enum:linTransEqThm} and
\ref{enum:linTransDepend} it is clear that
during the execution of the above algorithm,
the sum of the number of variables and
equations decreases. Therefore, the algorithm
terminates.

Let $I$ be the final value of $i$ after the execution of
the algorithm. Since the only way for our algorithm to
terminate is through step \ref{enum:linTransEqNonZeroDetVarCheck}
it holds that system $\system_I$
is an $e_I \times e_I$ system of linear equalities with
non-zero determinant (for $e_I \leq r$).
System $\system_I$ is obtained from system
$\system_0$ by possibly replacing some variables that correspond
to zero entries of the solution with zeros
and by possibly removing some equalities (that
have a linear dependence on others). So,
any solution of $\system_I$ is also a solution of
$\system_0$ and thus a solution of
$\system$. From the algorithm we have that
$\myvec{x^{*,I}}$ is a solution of
$\system_I$. Since $\system_I$ has a non-zero
determinant Cramer's rule can 
be applied. Hence, 
the vector $\myvec{x^{*,I}}$
is the unique solution of system $\system_I$.
Let $x^{*,I}_i$ be an entry of
$\myvec{x^{*,I}}$. Entry
$x^{*,I}_i$ is equal to the following rational number:
\[
\frac{
\begin{vmatrix}
a_{11} & \ldots & a_{1e_I}\\
& \ddots & \\
a_{e_I1} & \ldots & a_{e_Ie_I}
\end{vmatrix}
}{
\begin{vmatrix}
b_{11} & \ldots &  b_{1e_I}\\
& \ddots & \\
b_{e_I1} & \ldots & b_{e_Ie_I}
\end{vmatrix}
}~,
\]
where all the $a_{ij}$ and $b_{ij}$ are integers that
appear in the original system $\system$. By properties
of the determinant we know that the numerator
and the denominator
of the above rational number
will each be at most equal to
$r! \cdot (2^l - 1) ^ r$. So we have that:
\begin{align*}
|x^{*,I}_i| & \leq 2 \cdot \big 
(\log_2 (r! \cdot (2^l - 1) ^ r)  + 1 \big ) 
& \Longrightarrow\\
|x^{*,I}_i| & \leq 2 \cdot 
\big ( \log_2 (r^r \cdot 2^{l \cdot r})  + 1 \big ) 
& \Longrightarrow\\
|x^{*,I}_i| & \leq 2 \cdot 
\big ( r \cdot \log_2 (r) + l \cdot r  + 1 \big )~.
\end{align*}
As we already mentioned the final
vector $\myvec{x^{*,I}}$ is a solution of the original linear
system $\system$. 
We also have
that all the entries of $\myvec{x^{*,I}}$ are non-negative,
at most $r$ of its entries are positive and
the size of each entry of $\myvec{x^{*,I}}$ is bounded by
$2 \cdot ( r \cdot \log_2 r + r \cdot l + 1)$. 
So, $\myvec{x^{*,I}}$ is the desired
vector $\myvec{y^*}$.
\end{proof}

\section{Non-Iterated Probabilistic Logics}

In \Fref{subsec:non-iter-sem} we define the
semantics for non-iterated probabilistic logics.
In \Fref{subsec:non-iter-smp} we prove a small
model property
and in \Fref{subsec:non-iter-compl} we present a conditional complexity upper bound
for these logics. 
\label{sec:non-iter}

\subsection{Semantics}

\label{subsec:non-iter-sem}

The non-iterated probabilistic logic over $\logic$
is the logic $\PL$. The language of $\PL$
is defined by adding
non-nested probabilistic operators to the language
$\lan$.  
Formally, $\lanPL = \lanPL' \cup \lan$, where
$\lanPL'$ is described by the following grammar:
\[
A :: = P_{\geq s} \alpha ~|~ \lnot A ~|~ A \land A~,
\]
where $s \in \Rat \cap [0,1]$ 
and $\alpha \in \lan$. Recall that by definition, the basic formulas of $\PL$
are the formulas of $\lanPL$ that do not have $\lnot, \land$ or a probabilistic
operator as a top connective. Hence we have that
$\basis(\lanPL) = \basis(\lan) \supseteq \Prop$, i.e.
$\lanPL$ has the same basic formulas as $\lan$.
The intended meaning of the formula
$P_{\geq s} \alpha$ is that
``the probability of truthfulness for 
$\alpha$ is at least $s$''. For $\lanPL$,
we assume the same abbreviations as for $\lan$.
The operator $P_{\geq s}$ is assumed to have greater 
precedence than all the connectives of
$\lan$. We also define the following syntactical
abbreviations:
\begin{align*}
P_{< s} \alpha & \equiv \lnot P_{\geq s} \alpha~; \\
P_{\leq s} \alpha & \equiv  P_{\geq 1 - s} \lnot 
\alpha~; \\
P_{> s} \alpha  & \equiv \lnot P_{\leq s} 
\alpha~; \\
P_{= s} \alpha & \equiv P_{\geq s} \alpha \land 
P_{\leq s} \alpha~.
\end{align*}
We use capital Latin letters like $A, ~B, ~C, ~\ldots$ for
members of $\lanPL'$ possibly primed or with
subscripts. 

\begin{remark}
In the literature non-iterated logics either contained~\cite{ograma09} or did not
contain~\cite{komaogst} the formulas of the base logic. In this paper we opted
for the first choice. This makes our approach more uniform since all of our
logics are extension of classical propositional logic. We have to point out that
as far as decidability and complexity is concerned, both approaches are practically
the same: if the language of the non-iterated probabilistic logic contains formulas
of the base logic, then we simply have to use the decidability algorithm for the
base logic too.
\end{remark}

A model for $\PL$ is a probability
space where the events (also called worlds) are
models for $\logic$.
In order to determine the probability of truthfulness for
an $\lan$-formula $\alpha$ in such a
probability space we have to find the
measure of the set containing all
$\lan$-models that satisfy $\alpha$. More formally, we have the following:

\begin{definition}[$\PL$-Model]
Let $M = \langle W, H, \allowbreak \mu, v 
\rangle$ where
\begin{itemize}[topsep = 0em]
\item
$\langle W, H, \mu \rangle$ is a probability
space~;
\item
$v$ is a function that assigns an evaluation
to every $w$ in $W$. 
We write $v_w$ instead of $v(w)$.
\end{itemize}
$M$ is a $\PL$-model if
$[\alpha]_M \in H$ for every
$\alpha \in \lan$, where
\[
[\alpha]_M =
\{ w \in W ~|~ v_w (\alpha) = \true\}~.
\]
We will drop the subscript $M$, i.e.~we 
will simply write $[\alpha]$, if this causes no confusion. 
\end{definition}

\begin{definition}[Truth in a $\PL$-model]
Let $M = \langle W, H, \mu, v \rangle$ be a
$\PL$-model. The truth
of $\lanPL'$-formulas that
have a probabilistic operator as their top-connective
is defined as follows (the formulas with top-connectives
$\lnot$ and $\land$ are treated classically):
\[
M \models P_{\geq s} \alpha \Longleftrightarrow 
\mu([\alpha]_M) \geq s~.
\]
Also if $\alpha$ is an $\lan$-formula then
\[
M \models \alpha \Longleftrightarrow [\alpha]_M = W~.
\]
\end{definition}
We observe that in order to present the formal semantics of
a non-iterated probabilistic logic, we have to define the
basic formulas and explain how the evaluation behaves on
them.

\subsection{Small Model Property}

\label{subsec:non-iter-smp}

In this subsection we show that if $A \in \lanPL'$ is satisfiable
then it is satisfiable in a model that satisfies the following properties:
\begin{itemize}
\item
the number of worlds and the probabilities assigned to them have size polynomial in the size of $A$
\item
the evaluations assigned to every world depend only on the subformulas of $A$.
\end{itemize}
After we have established this result, it
is easy to obtain the upper complexity bound for $\PL$:
we can simply guess the model
in polynomial time and then, with the help of
some oracles, verify that it satisfies $A$.

First we need some definitions.
The set of subformulas of some formula
$A$, represented as $\subf(A)$, is defined as usual.
The size of $A$, represented as $|A|$,
is the number of symbols needed to write $A$.
In order to compute $|A|$, the size of every
probabilistic operator counts as one.
For example, $|\lnot P_{\geq \frac{1}{5}} p| = 3$.
For $A \in \lanPL$ we define
\[
||A|| := \max \big  \{ |s| ~ \big |~ P_{\geq s } \alpha 
\in \subf(A) \big \}~.
\]

\begin{definition}[Conjunctions of Positive
and Negative Basic Formulas]
Let $A \in \lanPL$. The set
of conjunctions
of positive and negative basic formulas
of $A$ is the following set:
\[
\cpnb(A) = \left \{ a ~ \Bigg |~ a\text{ is of the form }\bigwedge_{B \in \subf(A) \cap \basis(\lan)} \pm B \right \} ~,
\]
where $\pm B$ denotes either $B$ or $\lnot B$.
The acronym $\cpnb$ stands for $\mathsf{c}$onjunction of
$\mathsf{p}$ositive and
$\mathsf{n}$egative $\mathsf{b}$asic formulas.
If $a \in \cpnb(A)$ for some $A$ and
there is no danger of confusion we may say
that $a$ is $\cpnb$-formula. We use the 
lower-case Latin letter $a$ for $\cpnb$-formulas, 
possibly with subscripts.
\end{definition}
Let $A$ be of the form
$\bigwedge_{i} B_i$ or of the form
$\bigvee_{i} B_i$. Then
$C \in A$ means that for some $i$, $B_i \equiv C$.

\Fref{thm:smp_PL} proves the announced small model property.
It is an adaptation of
the small model Theorem 2.6 of~\cite{fahame90}.
In \cite{kokkinis16} the
proof of the small property unnecessarily depends
on the completeness theorem for $\PJ$. In  \Fref{thm:smp_PL} we remedy
this mistake. 
\begin{theorem}[Small Model Property for $\PL$]
\label{thm:smp_PL}
Let $A \in \lanPL'$. If $A$ is $\PL$-satisfiable then it is satisfiable
in a $\PL$-model $M = \langle W, H, \mu, v \rangle$ such that
\begin{enumerate}[label=(\arabic*)]
\item
\label{enum:smpWorlds}
$|W| \leq |A|$;
\item
$H = \powerset (W)$;
\item
For every $w \in W$,
$\mu (\{ w \})$ is a non-negative
rational number with size at most
\[
2 \cdot \big ( |A| \cdot || A || + |A| \cdot \log_2 (|A|) + 1\big );
\]
\item
\label{enum:smpAtoms}
For every $a \in \cpnb(A)$, there exists
at most one $w \in W$ such that $\ext{v}_w(a) = \true$.
\end{enumerate}
\end{theorem}

\begin{proof}
Let $A$ be satisfiable in some $\PL$-model.
We divide the proof in two parts: 
\begin{itemize}
\item 
we show that the satisfiability of $A$ 
implies that a linear system $\system$ is satisfiable;
\item 
we use a solution of $\system$ to define the model $M$ for $A$
that satisfies the properties
\ref{enum:smpWorlds}--\ref{enum:smpAtoms}.
\end{itemize}
\paragraph{\textbf{Finding the Satisfiable Linear System.}}
Let $R$ be some $\PL$-model. By
propositional reasoning we can show that
\begin{equation}
\label{eq:satEquiv}
R \models A \Longleftrightarrow R \models
\bigvee_{i=1}^{K} 
\bigwedge_{j=1}^{l_i} P_{\op_{ij} s_{ij}} 
\left (\alpha^{ij} \right )~.
\end{equation}
for some $K$ and $l_i$'s, such that for each
$i$ and for each $j$,
$\op_{ij} \in \{ \geq, <\}$ and $\alpha^{ij}$
is a disjunction of elements of $\cpnb(A)$.
Since $A$ is satisfiable, Eq.~\eqref{eq:satEquiv} 
implies that there exists
a $\PL$-model $M' = \langle W', H', \mu', v' \rangle$
and some $ 1 \leq i \leq K$ such that
\begin{equation}
\label{eq:SMPdisjSat}
M' \models
\bigwedge_{j=1}^{l_i} P_{\op_{ij} s_{ij}} \left ( \alpha^{ij} \right )~.
\end{equation}
Let $\cpnb(A) = \{ a_1 , \ldots , a_n \}$.
For every $k \in \{1, \ldots, n\}$ we define
\begin{equation}
\label{eq:XeqMu}
x^*_k = \mu'([a_k]_{M'})~.
\end{equation}
In every world of $M'$ some atom of $A$ must hold.
Thus, we have
\begin{equation}
\label{eq:SMPmuOne}
\mu' \left ( \bigcup^n_{k=1} [a_k]_{M'} \right ) = 1~.
\end{equation}
All the $a_k$'s belong to $\cpnb(A)$, so
for all $k$, $k' \in \{ 1, \ldots, n \}$, we have
\begin{equation}
\label{eq:SMPdisj}
k \neq k' \Longrightarrow
[a_k]_{M'} \cap [a_{k'}]_{M'} = \emptyset~.
\end{equation}
By Eqs.~\eqref{eq:XeqMu},\eqref{eq:SMPmuOne},\eqref{eq:SMPdisj} and 
the additivity of $\mu'$ we get
\begin{equation}
\label{eq:SMPsumOne}
\sum^n_{k=1} x^*_k= 1~.
\end{equation}
Let $j \in \{ 1, \ldots , l_i \}$. From Eq.~\eqref{eq:SMPdisjSat} we get
$M' \models P_{\op_{ij} s_{ij} } \big (\alpha^{ij} \big )$.
This implies that
$\mu' ( [\alpha^{ij}]_{M'})\ \op_{ij} \ s_{ij}$,
i.e.
\[
\mu' \left ( \left [\bigvee_{a_k \in \alpha^{ij}} a_k 
\right ]_{M'} \right )\  
\op_{ij} \ s_{ij},
\] 
from which we can show that
\begin{equation*}
\mu' \Bigg ( \bigcup_{ a_k \in \alpha^{ij}}
[ a_k ]_{M'} \Bigg )\ \op_{ij} \ s_{ij}~.
\end{equation*}
By Eq.~\eqref{eq:XeqMu}, \eqref{eq:SMPdisj} and the additivity of $\mu'$ we have that
\[
\sum_{a_k \in \alpha^{ij}} x^*_k \  \op_{ij} \ s_{ij}~.
\]
So we have that
\begin{equation}
\label{eq:sumOpS}
\text{for every } j \in \{ 1, \ldots , l_i\},
\sum_{a_k \in \alpha^{ij}} x^*_k \  \op_{ij} \ 
s_{ij}~.
\end{equation}

By Eqs. \eqref{eq:SMPsumOne}
and \eqref{eq:sumOpS} it is clear that the
vector $\myvec{x^*} = x^*_1, \ldots, x^*_n$ is a
non-negative solution of a linear system,
call it $\system$. By \Fref{thm:lin_ineq_eq_thm} 
we have that there exists
a vector $\myvec{y^*} = y^*_1,
\ldots, y^*_n$, with non-negative entries,
that is a solution of $\system$
and has at most $N$ entries (strictly) positive, where $0 < N \leq |A|$.
Without loss of generality we assume that
$y^*_1, \ldots, y^*_N$ are the positive entries of 
$\myvec{y^*}$.
Since every $x^*_k$ corresponds to a $\cpnb$-formula of
$A$ we can associate every positive
$y^*_k$ with the satisfiable atom $a_k$.

\paragraph{\textbf{Defining the Model $M$ for $A$.}}
The quadruple $M = \langle W, H, \mu, v \rangle$ is defined
as follows:
\begin{enumerate}[label=(\alph*), topsep = 0em]
\item
$W = \{ w_1 , \ldots, w_N\}$, for some $w_1, \ldots, w_N$;
\item
$H = \powerset (W)$;
\item
For all $V \in H$,
\[
\mu(V) =  \sum_{w_k \in V} y^*_k;
\]
\item
Let $i \in \{ 1, \ldots, N\}$.
$v_{w_i}$ is an evaluation that satisfies $a_i$.
\end{enumerate}
By using the fact that each $y^*$ is a solution of
$\system$ we can show that $M$ is a $\PL$-model. We will now prove 
the following statement:
\begin{equation}
\label{eq:wKaK}
(\forall 1 \leq k \leq N)  \left [ w_k \in \left [\alpha^{ij} \right ]_M 
\Longleftrightarrow a_k \in \alpha^{ij} \right ]~.
\end{equation}

Let $k \in \{1, ~\ldots~, ~N\}$.
We prove the two directions of Eq.~\eqref{eq:wKaK} separately.

$(\Longrightarrow)$
Assume that $w_k \in [\alpha^{ij}]$. This means that
$v_{w_k} (\alpha^{ij}) = \true$.
Assume that $a_k \notin \alpha^{ij}$. Then, since
$\alpha^{ij}$ is a disjunction of some
atoms of $A$, there must
exist some $a_{k'} \in \alpha^{ij}$, with $k \neq k'$, such that
$v_{w_k} (a_{k'}) = \true$. 
However, by definition we have that $v_{w_k} (a_k) = \true$.
But this is a contradiction, since
$a_k$ and $a_{k'}$ are different
atoms of the same formula, which means
that they cannot be satisfied in the same
evaluation. Hence, $a_k \in \alpha^{ij}$.

$(\Longleftarrow)$
Assume that $a_k \in \alpha^{ij}$. We know that $v_{w_k}(a_k) = \true$, which implies that
\[
v_{w_k} (\alpha^{ij}) = \true, \text{ i.e. } w_k
\in \left [\alpha^{ij} \right ]_M~.
\]

Hence, Eq.~\eqref{eq:wKaK} holds.
Now, we will prove the following statement:
\begin{equation}
\label{eq:SMPsat}
\big ( \forall 1 \leq j \leq l_i \big ) \big [
M \models P_{\op_{ij} s_{ij}} \alpha^{ij} \big ]~.
\end{equation}

Let $j \in \{ 1, \ldots, l_i \}$.
It holds
\begin{align*}
M & \models P_{\op_{ij} s_{ij}} (\alpha^{ij})
& \Longleftrightarrow\\
\mu ([& \alpha^{ij}]_M) ~\op_{ij}~ s_{ij}
& \Longleftrightarrow\\
\sum_{w_k \in [\alpha^{ij}]_{M}}& y^*_k ~\op_{ij}~ s_{ij} & 
\stackrel{\text{Eq.~\eqref{eq:wKaK}}}{\Longleftrightarrow} \\
\sum_{a_k \in \alpha^{ij}}& y^*_k ~\op_{ij}~ s_{ij}~.
\end{align*}
The last statement holds
because $\myvec{y^*}$ is a solution of $\system$. 
Thus, Eq.~\eqref{eq:SMPsat} holds.

By Eq.~\eqref{eq:SMPsat} we have that $M \models 
\bigwedge^{l_i}_{j=1} P_{\op_{ij} s_{ij}} (\alpha^{ij})$, which 
implies that
\[
M \models 
\bigvee_{i=1}^{K}\bigwedge^{l_i}_{j=1} P_{\op_{ij} s_{ij}}
(\alpha^{ij}),
\]
which, by Eq.~\eqref{eq:satEquiv}, implies that $M \models A$.

So, we have that each $w_i$ corresponds to
one satisfiable $a_i$ and also that $\mu(\{w_i\}) = y^*_i$.
Since the number of positive $y^*_i$'s is at most $|A|$ and the size
of every positive $y^*_i$ is at most $2 \cdot (|A| \cdot ||A|| + |A| \log_2(|A|)+1)$,
we have that $M$ is, indeed, the model in question.
\end{proof}

The small model property
shows that the formula is satisfied in a structure with small number of
worlds, small probabilities assigned to each world and that in every world of
this structure a unique $\cpnb$-formula holds. The following Lemma shows
these $\cpnb$-formulas that hold in the worlds practically define
evaluations for the formula that is tested for satisfiability.

\begin{lemma}
\label{lem:atomEval}
Let $\alpha \in \lan$, let $v_1, v_2$ be two evaluations
and assume that, for every
basic formula $\beta$ that appears in $\alpha$
\[
v_1(\beta) = v_2(\beta)~.
\] 
Then we have
\[
\ext{v}_1 (\alpha) = \ext{v}_2 (\alpha)~.
\]
\end{lemma}

\subsection{Complexity Bounds}

\label{subsec:non-iter-compl}

In this subsection we obtain the conditional upper bound for $\PLSat$.
The upper bound follows from the fact that for a given $\PL$-formula $A$,
we can guess a small model for it and then verify that this
model indeed satisfies $A$. 

As a first step we need
the following Lemma which can be proved by an easy
induction on the complexity of the formula.

\begin{lemma}
\label{lem:atomSat}
Let $\alpha \in \lan$ and let $a \in 
\cpnb(\alpha)$. Let $v$ be an evaluation and
assume that $\ext{v}(a) = \true$. The decision
problem
\begin{center}
does $\ext{v}$ satisfy $\alpha$?
\end{center}
belongs to the complexity class $\mathsf{P}$.
\end{lemma}
Now we are ready to prove the upper complexity
bound for $\PLSat$.
\begin{theorem}
\label{thm:upper_bound_pl}
Assume that the satisfiability problem for
$\cpnb$-formulas in the logic $\logic$
belongs to the complexity class $\cclass$.
Then $\PLSat \in \np^{\cclass}$.
\end{theorem}

\begin{proof}
Let $A \in \lanPL$ and let $\algo$ be the
$\cclass$-algorithm that can test $\cpnb$-formulas
in $\logic$ for satisfiability. If $A \in \lan$ then $A$ is a equivalent to a
disjunction of $\cpnb$-formulas (i.e. $A$ can be seen as a formula in
disjunctive normal form). So, we guess one of these formulas and
using $\algo$ verify that it is satisfiable in polynomial time using
\Fref{lem:atomSat}. Of course this can be done
in nondeterministic polynomial time using a $\cclass$-oracle.

For the rest of the proof we assume that $A \in \lanPL'$, i.e. that $A$ contains
probabilistic operators.
A non-deterministic algorithm that tests $A$ for
satisfiability can simply guess a model for
$A$ that satisfies
conditions~\ref{enum:smpWorlds}--\ref{enum:smpAtoms}
that appear in the statement of \Fref{thm:smp_PL}.
We present a non-deterministic algorithm
that performs this guess and we evaluate its
complexity.

\paragraph{\textbf{Algorithm.}} We guess $n$ elements of $\cpnb(A)$,  call them
$a_1, \ldots , a_n$, and we also choose
$n$ worlds, $w_1, \ldots, w_n$, for $n \leq |A|$.
Using $\algo$ we can verify that for each
$i \in \{ 1, \ldots, n\}$ there exists an
evaluation  $\ext{v}_i$ such that $\ext{v}_i(a_i) = \true$. 
We define $W = \{ w_1 , \ldots, w_n\}$ and
for every $i \in \{ 1, \ldots, n\}$ we set $v_{w_i} = v_i$. Since we are only interested in
the satisfiability of basic
formulas that appear in
$A$, by \Fref{lem:atomEval}, the choice of the $v_{w_i}$
is not important (as long as $v_{w_i}$ satisfies
$a_i$).
We assign to every $\mu(\{w_i\})$ a rational number
with size at most
\[
2 \cdot \big ( |A| \cdot ||A|| + |A| \cdot \log_2 
(|A|) + 1 \big )~.
\]
We set $H = \powerset(W)$ and
for every $V \in H$ we set
\[
\mu(V) = \sum_{w_i \in V} \mu (\{w_i\})~.
\]
It is then straightforward to see that
conditions~\ref{enum:smpWorlds}--\ref{enum:smpAtoms}
that appear in the statement of \Fref{thm:smp_PL} hold.

Now we have to verify that our guess is correct, i.e. that
$M \models A$.
Assume that $P_{\geq s} \alpha$ appears in $A$.
In order to see whether $P_{\geq s} \alpha$ holds
we need to calculate the measure of the set $[\alpha]_M$
in the model $M$. The set $[\alpha]_M$ will contain
every $w_i \in W$ such that $v_{w_i}(\alpha) =
\true$.
Since $v_{w_i}$ satisfies an atom of $A$ it also
satisfies an atom of $\alpha$. So, by 
\Fref{lem:atomSat}, we can check whether $*_{w_i}$
satisfies $\alpha$ in polynomial time.
If $\sum_{w_i \in [\alpha]_M} 
\mu(\{w_i\}) \geq s$ then we replace $P_{\geq s} \alpha$
in $A$
with the truth value $\true$, otherwise with
the truth value $\false$.
We repeat the above procedure for every formula
of the form $P_{\geq s} \alpha$ that 
appears in $A$. At the end we have a formula that is
constructed only from the connectives $\lnot$, $\land$
and the truth constants $\true$ and $\false$. Obviously,
we can verify in polynomial
time that the formula is true.
This, of course, implies that $M \models A$.

\paragraph{\textbf{Complexity Evaluation.}}
All the objects that are guessed in our
algorithm have size that is polynomial in the size
of $A$.
Also the verification phase of our algorithm can
be made in polynomial time. Furthermore
checking whether an element of
$\cpnb(\alpha)$ is satisfiable 
is possible with a $\cclass$-oracle.
Thus, our $\PLSat$ belongs to the class
${\np}^{\cclass}$.  
\end{proof}

\section{Iterated Probabilistic Logics}

\label{sec:iter}

\subsection{Semantics}

The iterated probabilistic logic over $\logic$
is the logic $\PPL$ (the two $\mathsf{P}$'s stand for 
the iterations of the probability operator).
The language of $\PPL$,
$\lanPPL$, is defined by adding nested probabilistic
operators to the language $\lan$.
Formally, $\lanPPL$ is defined by the
following grammar:
\[
A ::= b ~|~ \lnot A ~|~ A \land A ~|~ 
P_{\geq s} A~,
\]
where $s \in \Rat \cap [0,1]$ and $b \in \basis(\lanPPL) \supseteq \basis(\lan)$.
We will use upper-case latin letters
like $A, B, C, \ldots$ for members of $\lanPPL$.

Models for $\PPL$ are probability
spaces where the worlds contain 
evaluations and probability
spaces (so that we can deal with iterated probabilities).
Formally, we have

\begin{definition}[$\PPL$-Model and Truth in a $\PPL$-Model]
Assume that $M = \langle U, W, H, \mu , v \rangle$ where:
\begin{enumerate}
\item
$U$ is a non-empty set of objects called worlds;
\item
for every $w \in U$,
\begin{center}
$\langle W_w, H_w, \mu_w\rangle$ is
a probability space with $W_w \subseteq U$
and $v_w$ is an evaluation.
\end{center}
\end{enumerate}
Truth in $M$ is defined as follows (the connectives $\lnot$ and
$\land$ are treated classically):
\begin{align*}
M, w \models A &\quad\Longleftrightarrow\quad v_w (A) = \true \quad\text{for $A \in \basis(\lanPPL)$}~;\\
M, w \models P_{\geq s} A &\quad \Longleftrightarrow\quad
\mu_w \big ( [A]_{M,w} \big ) \geq s.
\end{align*}
$M$ is a $\PPL$-model if for every $A \in \lanPPL$ and every $w \in U$, $[A]_{M,w} \in H_w$, where
\[
[A]_{M,w} = \{ u \in W_w ~|~ M,u \models A \}.
\]
\end{definition}

We observe that, as in the non-iterated case, in order to formally define the semantics
of an iterated probabilistic logic  it suffices to define the basic formulas of $\PPL$ and how
the evaluation behaves on them.

\subsection{Complexity Bounds}

In this section we obtain complexity bounds for $\PPLSat$.
The upper bound
is obtained via a tableaux procedure, which
resembles the tableaux procedure for modal 
logic~\cite{halmos92}.
The idea for obtaining this upper bound
for a modal logic that contains probabilistic
operators similar to ours was sketched
in \cite[Theorem 4.5]{faginH94}, but no complete formal proof was given there.
The lower bound is obtained by drawing a reduction from
modal logic $\D$~\cite{halmos92}.

\subsubsection{The Upper Bound}

As a first step we need the following definition:

\begin{definition}[Conjunctions of Positive
and Negative Formulas]
For
$A_1,$ $\ldots,$ $A_n \in \lanPPL$, we define the following set:
\[
\cpnf(A_1, \ldots, A_n) = \left \{ a ~ \Bigg |~ a\text{ is of the form }\bigwedge^n_{i=1} \pm A_i \right \} ~.
\]
The acronym $\cpnf$ stands for 
$\mathsf{c}$onjunction of $\mathsf{p}$ositive and
$\mathsf{n}$egative $\mathsf{f}$ormulas. 
As for the $\cpnb$-formulas
we will use the possibly primed or subscripted
lower-case Latin letter $a$ for
$\cpnf$-formulas.
If $a \in \cpnf(A_1, \ldots, A_n)$ for some 
$A_1, \ldots, A_n$ and
there is no danger of confusion, we may say
that $a$ is a $\cpnf$-formula.
\end{definition}

We now present the announced tableaux method.
Our tableaux are trees where the nodes  are
formulas prefixed with world and truth signs.
So, the node $w~\true~A$ ($w~\false~A$) intuitively
means that formula $A$ is true 
(respectively false) at world $w$ of some model.
The root of a tableau contains the formula that is
tested for satisfiability. The tableaux rules
are presented in \Fref{tab:tableaux}.
\begin{table}
{
\centering
\renewcommand*{\arraystretch}{3}
\begin{tabular}{|c|}
\hline
\AxiomC{$w~\true~\lnot A$}
\myLL{\notT}
\UnaryInfC{$w~\false~A$}
\DisplayProof
\hspace{0.5em}
\AxiomC{$w~\false~\lnot A$}
\myLL{\notF}
\UnaryInfC{$w~\true~A$}
\DisplayProof
\hspace{0.5em}
\AxiomC{$w~\true~A \land B$}
\myLL{\andT}
\UnaryInfC{$w~\true~A$}
\noLine
\UnaryInfC{$w~\true~B$}
\DisplayProof
\hspace{0.5em}
\AxiomC{$w ~\false~A \land B$}
\myLL{\andF}
\UnaryInfC{$w ~\false~A~|~w ~\false~B$}
\DisplayProof\\
\AxiomC{$p_{w_{ij}}$}
\myLL{\tabProb}
\UnaryInfC{$w.1 ~\true~ a_1 ~|~ \cdots ~|~ w.n ~\true~ a_n $}
\DisplayProof \\\hline
\end{tabular}
\caption{\label{tab:tableaux} The Tableaux Rules}
}
\end{table}
The first line
consists of the propositional rules and
the second line of the probabilistic rule
$\tabProb$. A separator, i.e. the symbol
``$|$'', in the result of the
rule means that the formulas in the conclusion
belong to distinct branches. So, only
the rules $\tabProb$ and $\andF$ create new
branches; the other rules simply add formulas to the
branch where the premise belongs. 

Every propositional rule gives simpler conditions for satisfiability: if the premise
is satisfiable then at least one of the results has to be satisfiable too.
The function of rule $\tabProb$ is more complicated and requires some explanation.
Rule $\tabProb$ is the only rule that
creates new worlds. So, formulas that belong to a path
between two applications of rule $\tabProb$ are marked
with the same world. Therefore, we can define the notion
of a \emph{world path}. A world path is a shortest path in a
tableau that starts either from the root or from
a result of an application of the rule $\tabProb$ and
ends either in a leaf or at a premise of an
application of the rule $\tabProb$. We assume that the root, each one of the leaves and
each application of rule $\tabProb$ are marked with unique natural
numbers. In the following we will refer to the numbers assigned to
the root, the leaves, or the applications of rule $\tabProb$ as points.
Due to the fact that all nodes in a world path are marked with the
same world and lay between two points, we can represent a world path as $p_{w_{ij}}$ where $w$
is the world prefix of the formulas in the world path and $i,j$ are the points.
If formula $A$ appears in world path $p_{w_{ij}}$ (prefixed either with $\true$ or with $\false$) 
we write $A \in p_{w_{ij}}$. Now, the function of rule $\tabProb$ in 
\Fref{tab:tableaux} can be explained in full 
detail: it is applied in the 
world path $p_{w_{ij}}$ and it creates several new
branches. Each
of these branches is marked with a new
world symbol $w.k$ ($1 \leq k \leq n$). Assume that
\[
\{ B_1, \ldots, B_m\} = \left \{ B ~|~ P_{\geq s} B
\in p_{w_{ij}} \right \}~.
\]
Then the formulas $a_l$ appearing in the 
result of rule $\tabProb$ are defined as follows:
\[
\{ a_1, \ldots, a_n \} = \cpnf(B_1, \ldots, B_m)~.
\]
It is not difficult to see that each $B_k$ is equivalent to a
disjunction of some $a_l$'s. Hence the world path $p_{w_{ij}}$ imposes
probabilistic conditions of the form $P_{\geq s} \bigvee_l a_l$ or
$P_{< s} \bigvee_l a_l$, which using the fact that the $a_l$'s cannot
hold in the same world, translate to a condition like
``the sum of the measures of some $a_l$'s is at least or less than $s$''.
So, even if in every result of the rule $\tabProb$ a different $\cpnf$-formula has
to hold, each of these formulas will be assigned a (possibly) different
probability (which could be $0$)
and these probabilities will have to satisfy some linear conditions imposed
by the formulas in $p_{w_{ij}}$. Thus, while in the propositional rules we have
the property ``if the premise holds, at least one of the results has to hold'', in rule
$\tabProb$ we have ``if the premise holds (i.e. if the conjunction of all
the formulas that appear in the world path of the premise hold) then each one of the results
have to hold with (possibly) different probability and some sums of these probabilites have to satisfy
some linear conditions''.

A world path is called $P$-open if there is
some $P_{\geq s} B \in p_{w_{ij}}$. Otherwise it is
called $P$-closed. Let $p_{w_{ij}}$ be a $P$-open
world path that ends in the application
of rule $\tabProb$ $j$.
All the world paths that start from a
result of $j$ are called the \emph{children} of $p_{w_{ij}}$.
To make the presentation simpler we will
use the nodes with the subscription of points as expressions in the metalanguage. 
So, we simply
write ``$w_{ij}~\true~A$'' instead of the phrase ``the node $w~\true~A$
appears in the tableau in the world path between the points $i$ and $j$''.
Observe that the simpler notation $w~\true~A$ is ambiguous.
The reason is that the application of rule $\andF$ creates two branches, where the formulas
are prefixed with the same world sign. This means that the same formula
may occur with the same world sign in different places in the tableau tree.
We also have to point out
that the points are only
necessary for defining the worlds of the model that will be obtained from the tableaux. 
Even if some points appear
in the premise of rule $\tabProb$, they play no role for constructing the tableaux. 

For every $p_{w_{ij}}$ we define  the following $\PPL$-formulas:
\begin{align*}
F_{p_{w_{ij}}} & = \bigwedge_{w_{ij}~\true~C} C \land
\bigwedge_{w_{ij}~\false~C} \lnot C~;\\
B_{p_{w_{ij}}} & = \bigwedge_{w_{ij}~\true~C, C \in \basis(\lanPPL)} 
C \land \bigwedge_{w_{ij}~\false~C, C \in \basis(\lanPPL)} 
\lnot C~.
\end{align*}
The intuition behind these definitions is that $F_{p_{w_{ij}}}$ is a conjunction of
all the positive and negative formulas (hence the letter $F$) that apppear in a world path, whereas
$B_{p_{w_{ij}}}$ is a conjunction of all the positive and negative basic formulas (hence
the letter $B$) that appear in a world path. Observe that $B_{p_{w_{ij}}}$ is always
a $\cpnb$-formula. Also if $F_{p_{w_{ij}}}$ is $\PPL$-satisfiable
this implies that there is an evaluation that
satisfies $B_{p_{w_{ij}}}$. Let $p_{w_{ij}}$ be a premise of an
application of rule $\tabProb$. When the conjunction
of formulas in  $p_{w_{ij}}$ hold (i.e. when formula $F_{p_{w_{ij}}}$ holds) some
probabilistic conditions are imposed (which are translated to linear conditions in
the probabilities assigned to the children of $p_{w_{ij}}$) and also some non-probabilistic
conditions have to hold (which are translated to satisfiability conditions imposed
to basic formulas appearing in $p_{w_{ij}}$, i.e. to the formula $B_{p_{w_{ij}}}$).

The tableau for some $A \in \lanPPL$ is a tree
that is created as follows:

\begin{enumerate}
\item
Create the node $w~\true~A$ (this is the root of the tableau). Go to step \ref{enum:create_points}.
\item
\label{enum:create_points}
Assign to the root, each one of the leaves and every application of the rule $\tabProb$
a unique natural number. Go to step \ref{enum:tab_rule_prop}.
\item
\label{enum:tab_rule_prop}
Apply the propositional rules for
as long as possible.
If there exists a $P$-open
world path, go to
step \ref{enum:tab_rule_prob}.
Otherwise stop.
\item
\label{enum:tab_rule_prob}
Apply the rule $\tabProb$ to every 
$P$-open world path. Go to step~\ref{enum:tab_rule_prop}.
\end{enumerate}

Every time a tableaux rule is applied,
at least one operator ($\lnot$, $\land$ or $P_{\geq s}$)
is eliminated. This implies that the tableau for $A$ is a
finite tree. 

The goal of the tableaux procedure is to create a
model for some formula, if such a model exists. It is important for the reader to keep in
mind that the worlds of this model will not be the worlds assigned to each node, i.e. the
$w$'s. The actual worlds of the model will be the $w$'s, subscripted with points, i.e. the $w_{ij}$'s. 
So, it might help the reader to think of the $w$'s as pre-wordls, which can be instantiated
to several actual worlds, i.e. $w_{ij}$'s.

Now we are ready to prove the main theorem of this section.

\begin{theorem}
\label{thm:upper_bound_ppl}
Assume that the satisfiability problem for
$\cpnb$-formulas in the logic $\PPL$ belongs
to the complexity class $\cclass$.
Then $\PPLSat \in \pspace^{\cclass}$.
\end{theorem}

\begin{proof}
Let $A$ be the $\lanPPL$-formula that we want to 
test for satisfiability. Let $\algo$
be the $\cclass$-algorithm that decides the satisfiability
problem for $\cpnb$-formulas in logic $\PPL$.
We present an algorithm that decides whether
$A$ is satisfiable by traversing the tableau for
$A$ in a depth first fashion. Then, we prove the correctness
of the algorithm and analyse its complexity.

\paragraph{\textbf{Algorithm}.}

The goal of the algorithm is to traverse
the tableau for $A$ and decide which
world paths should be marked \emph{realizable}.
A realizable world path $p_{w_{ij}}$ contains all 
the formulas that are satisfied in world $w_{ij}$
of the model for $A$ (if our algorithm decides 
that such a model exists). On the other
hand, a world path that is not marked realizable
implies that the formulas in this path cannot be
satisfied in a $\PPL$-model. We execute the
following steps:

\begin{enumerate}
\item
\label{enum:choose_p_w}
If all the world paths have been examined then stop.
Otherwise, let $p_{w_{ij}}$ be the next (in depth first fashion) world path. Go to step
\ref{enum:p_closed}.
\item 
\label{enum:p_closed}
Mark $p_{w_{ij}}$ examined.
If $\algo$ fails in $B_{p_{w_{ij}}}$, do not mark it
realizable and go to step \ref{enum:choose_p_w}, else
if $p_{w_{ij}}$ is $P$-closed, mark $p_{w_{ij}}$
realizable and go to step \ref{enum:choose_p_w}, else
go to step \ref{enum:p_open}.
\item 
\label{enum:p_open}
Recall that by entering this step, we have selected an open world path $p_{w_{ij}}$,
such that $\algo$ succeeds in $B_{p_{w_{ij}}}$. Now we proceed as follows:
 elect at most $|A|$ rational
numbers (not necessarily different from each other) 
of size at most $2 \cdot \big ( |A| \cdot || A || + |A| \cdot \log_2 (|A|) + 1\big )$ 
from the interval $(0,1]$ and
assign each one of them to a child of $p_{w_{ij}}$.
Assign the number $0$ to the rest of $p_{w_{ij}}$'s children. 
Now for each $w~\true ~ P_{\geq s}B$ that appears in $p_{w_{ij}}$
we run the procedure $\verb|prob_test_pos|(P_{\geq s}B, p_{w_{ij}})$.
And for each $w ~ \false~P_{\geq s} B$ that appears in $p_{w_{ij}}$
we run the procedure $\verb|prob_test_neg|(P_{\geq s} B, p_{w_{ij}})$.
These procedures are defined as follows:

\begin{flushleft}
$\verb|prob_test_pos|(P_{\geq s}B, p_{w_{ij}})$\\
find all the children of $p_{w_{ij}}$
that contain $B$ prefixed with $\true$ and 
add the rational numbers assigned to them. 
If the sum is less than $s$ return failure. Otherwise return success.
\end{flushleft}

\begin{flushleft}
$\verb|prob_test_neg|(P_{\geq s}B, p_{w_{ij}})$\\
find all the children of $p_{w_{ij}}$
that contain $B$ prefixed with $\true$ and 
add the rational numbers assigned to them. 
If the sum is greater or equal to $s$ return failure. Otherwise return success.
\end{flushleft}

If at least one of the above executions of the
two procedures returns failure, then mark $p_{w_{ij}}$ as not realizable
and move to to step \ref{enum:choose_p_w}.

Let $X$ be the set of all children of $p_{w_{ij}}$
to which a positive rational
number is assigned. Run step \ref{enum:p_closed} of the algorithm to every member of $X$.
If there exists one member of $X$
where step \ref{enum:p_closed} of the algorithm fails then mark
$p_{w_{ij}}$ unrealizable.
Otherwise mark $p_{w_{ij}}$ realizable. Go to step \ref{enum:choose_p_w}.
\end{enumerate}

If, at the end of the algorithm there exists
a world path starting from the root,
that is marked realizable, return
``satisfiable''. Otherwise, return
``not satisfiable''.

\paragraph{\textbf{Correctness.}}
In order to prove our algorithm
correct it suffices to show that
for every world path $p_{w_{ij}}$:
\begin{equation}
\label{eq:real_sat}
p_{w_{ij}} \text{ is marked realizable} 
\Longleftrightarrow 
F_{p_{w_{ij}}} \text{ is }
\PPL\text{-satisfiable.}
\end{equation}

Let $p_{w_{ij}}$ be a world path. We prove the
two directions of \eqref{eq:real_sat}
separately:

\noindent ($\Longrightarrow$) We define the structure
$M= \langle U, W, H, \mu, v\rangle$ as follows:
\begin{align*}
U = \{ u_{kl} ~|~ & p_{u_{kl}} \text{ is marked realizable in the subtree }\\
& \text{of the tableau that has the first node 
of $p_{w_{ij}}$ as a root}\}~.
\end{align*}
And for every $u_{kl} \in U$, we have:
\begin{itemize}
\item
$W_{u_{kl}} = U$
and $H_{u_{kl}} = \powerset(W_{u_{kl}})$.
\item
For every $v_{mn}$, such that $p_{v_{mn}}$ is a child of
$p_{u_{kl}}$ we define
$\mu_{u_{kl}} (\{ v_{mn}\})$ to be the rational
number assigned to $p_{v_{mn}}$ (which can be $0$). 
For every $v_{mn} \in U$ 
such that $v_{mn}$ is not a child of $p_{u_{kl}}$, we set 
$\mu_{u_{kl}}(\{ v_{mn}\}) = 0$.
\item
for every $V \in H_u$: $\mu_{u_{kl}} (V) = 
\sum_{v_{mn} \in V} \mu_{u_{kl}} (\{ v_{mn}\})$.
\item
$v_{u_{kl}}$ is the evaluation that satisfies
$B_{p_{u_{kl}} }$. We know that such an evaluation exists
since $p_{u_{kl}} $ can be marked satisfiable only
if $\algo$ succeeds in $B_{p_{u_{kl}}}$.
\end{itemize}
Since for every $u_{kl} \in U$ the rational numbers assigned
to $p_{u_{kl}}$'s children were
selected in a way such that their sum equals $1$,
it is straightforward to show that
$M$ is a $\PPL$-model. We will now show that for
every $B \in \subf(A)$ and for every $u_{kl} \in U$:
\begin{equation}
\label{eq:tabl_model}
(u_{kl}~\true~B \Longrightarrow M,u_{kl} \models B)
\text{ and }
(u_{kl} ~\false~B \Longrightarrow M,u_{kl} \not\models B).
\end{equation}
We proceed by induction on $B$. The only
interesting case is when $B \equiv P_{\geq s} C$. 
Assume that $u_{kl} ~\true~P_{\geq s} C$. Then we
have
\[
\sum_{ \{ p_{v_{mn}} ~|~ p_{v_{mn}} \text{ is a child of } p_{u_{kl}}
\text{ and } v_{mn} ~\true~C \} } r_{v_{mn}} \geq s~,
\]
where $r_{v_{mn}}$ is the rational number that
is assigned to $p_{v_{mn}}$.
By i.h. and by the definition of $M$ we have
\[
\sum_{ \{ v_{mn} ~|~ v_{mn} \in W_{u_{kl}} \text{ and }
M,v_{mn} \models C \} } \mu_{u_{kl}} (\{v_{mn}\}) \geq s,
\]
which by the additivity of $\mu_{u_{kl}}$ gives us
\[
 \mu_{u_{kl}}([C]_{M,u_{kl}}) \geq s,
\]
i.e.
\[
M, u_{kl} \models P_{\geq s} C~.
\]
Exactly the same arguments prove the right
conjunct of \eqref{eq:tabl_model} and
this concludes the proof of the if direction of
\eqref{eq:real_sat}.

\noindent ($\Longleftarrow$) We prove the claim by induction on the
depth of $p_{w_{ij}}$ in the tableau for $A$.

If $p_{w_{ij}}$ ends in a leaf, then $p_{w_{ij}}$ is a $P$-closed world path,
which implies that $F_{p_{w_{ij}}} \equiv B_{p_{w_{ij}}}$.
The fact that
$F_{p_{w_{ij}}}$ is $\PPL$-satisfiable
implies that there is an evaluation
that satisfies $B_{p_{w_{ij}}}$. This implies
that $\algo$ is successful on $B_{p_{w_{ij}}}$,
so $p_{w_{ij}}$ is marked realizable.

Assume that $p_{w_{ij}}$ has children. Assume
that $F_{p_{w_{ij}}}$ is satisfiable in world $w_{ij}$ of the
$\PPL$-model $M$. Since $p_{w_{ij}}$ is $P$-open
there is a $K$ and some $C_{r}$'s and $s_{r}$'s 
such that $M, w \models \bigwedge^K_{r=1} P_{\op_r s_{r}} C_{i}$ for $\op_r \in \{\geq, <\}$. 
Let
\[
\cpnf( C_1, \ldots, C_K ) = \{ a_1, \ldots, a_m\}~.
\]
By propositional reasoning we can show that
$M, w \models \bigwedge^K_{r=1} P_{\op_r s_{r}} D_{r}$ where
every $D_r$ is equivalent to a disjunction of some $a_k$'s. 
Now we proceed as in
the proof of \Fref{thm:smp_PL}. We show that the
fact that the $P_{\op_r s_r} D_r$'s are satisfied in $w_{ij}$
implies that there is a linear system $\system$
which has as a solution a vector, every
entry of which corresponds to the
measure of some $a_k$. By \Fref{thm:lin_ineq_eq_thm} we
can show that at most $|A|$ entries of
a solution for $\system$ have to be positive. And
each of these entries has size at most
$2 \cdot \big ( |A| \cdot || A || + |A| \cdot 
\log_2 (|A|) + 1\big )$.
Recall that each $a_k$ is assigned to a child of
$p_{w_{ij}}$. So, there are at most $|A|$ positive
rational numbers that are assigned to children
of $p_{w_{ij}}$. The algorithm that traverses the tableau
for $A$ should be
able to find them. Then the algorithm
should be able to verify that these
rational numbers sum to $1$ and satisfy the nodes
that contain $P_{\op_r} C_r$ in $p_{w_{ij}}$. Also the fact that
$F_{p_{w_{ij}}}$ is satisfiable implies that
$\algo$ succeeds in $B_{p_{w_{ij}}}$. Furthermore
the $a_j$'s that correspond to
positive measures are satisfiable in $M$. This implies
that for every $p_{u_{jk}}$, such that $p_{u_{jk}}$ is a child of
$p_{w_{ij}}$ and a positive rational
number is assigned to $p_{u_{jk}}$, then $F_{p_{u_{jk}}}$ is
satisfied in $M$. By the induction hypothesis we have that the
children of $p_{w_{ij}}$ that correspond to positive measures
are marked realizable. We conclude that $p_{w_{ij}}$ is marked
realizable.

This concludes the proof of
\eqref{eq:real_sat}.

\paragraph{\textbf{Complexity Analysis.}}
We will show that our algorithm can decide
whether there exists a world path starting from $A$ that
should be marked realizable by using only
a polynomial number of bits and a 
$\cclass$-oracle. We observe that
whether $p_{w_{ij}}$ should be marked realizable
only depends on $p_{w_{ij}}$ and the subtree below it.
So, we can traverse the tableau tree in
depth first fashion, reusing space.
For every $p_{w_{ij}}$ we need a polynomial number
of bits to store $p_{w_{ij}}$ itself, the positive rational
numbers assigned to some of its children and the
$\cpnf$-formulas that are assigned to these children. 
We can verify (using the 
$\cclass$-oracle) that $B_{p_{w_{ij}}}$ is satisfiable
and that the probabilistic constraints in
$p_{w_{ij}}$ are satisfied. Then we can move to the
first child of $p_{w_{ij}}$ (among those of which a positive 
probability is assigned) and
repeat the same procedure. Clearly, once we
have that the first child of $p_{w_{ij}}$ is marked
realizable we do not need the space used for
this child any more. So, this space can be used
for the next child. We conclude that the
maximum number of information that we have to
store each time is at most equal to the
depth of the tree (which is polynomial on $|A|$)
times the number of bits needed to process
a single world path (which as we observed
is polynomial on $|A|$ again). We conclude
that our algorithm runs in polynomial space using a 
$\cclass$-oracle. 
\end{proof}

\subsubsection{The Lower Bound}

Before showing the reduction from modal logic $\D$
we observe that the tableau decision
procedure implies a small model property
for $\PPL$.

\begin{corollary}(Small Model Property for $\PPL$)
\label{cor:smp_ppl}
Let $A \in \lanPPL$ be satisfiable. Then $A$ is
satisfiable in a $\PPL$-model $M = \langle U, W,H, \mu,
v \rangle$, where
\begin{enumerate}
\item
$U \leq 2^{|A|}$;
\item
for each $w \in U$,
\begin{enumerate}
\item
$W_w = U$ and $H_w = \powerset(U_w)$;
\item
for every $V \in H_w: \mu_w(V) = \sum_{v \in V} \mu(\{v\})$.
\end{enumerate}
\end{enumerate}
\end{corollary}

\begin{proof}
The fact that $A$ is satisfiable implies that the
tableau procedure for $A$ succeeds. So, if we start
a tableau procedure with $w~\true~A$ in the root
we should find a realizable world path that starts
from $w~\true~A$. Then as in the proof of the proof of the "if" direction
of \eqref{eq:real_sat} we can construct a $\PPL$-model
for $A$ that satisfies the properties in the
statement of this corollary. Since the tableau for $A$ is
finite, the size of the model satisfying $A$ is
finite. 
\end{proof}

Now we proceed with some standard definitions
from modal logic. The language
of modal logic, $\lanMod$, is described
by the following grammar:
\[
A :: = p ~|~ \lnot A ~|~ A \land A ~|~ \Box A~,
\]
where $p \in \Prop$.
A Kripke model is structure $M = \langle W, R, v\rangle$
where $W$ is a non-empty set of worlds, 
$R \subseteq W \times W$ and $v$ is a function that assigns a
truth assignment (for classical propositional
logic) to every world in $W$. For each 
$w \in W$ we define the following set:
\[
R[w] = \{ u~|~ (w,u) \in R\}~.
\]
The semantics of $\lan_{\Box}$-formulas is given by the
following definition:

\begin{definition}[Truth in a Kripke model]
Let $M = \langle W, H, v \rangle$ be a Kripke model and let $A \in \lanMod$.
We define what it means for $A$ to hold in the world $w$ of $M$
(written as $M, w \models A$) by distinguishing the following cases:
\begin{description}
\item[$A \equiv \Box B$:]
\[
M, w \models \Box B \Longleftrightarrow (\forall u \in R[w]) [M, u \models B]
\]
\item[$A \equiv p \in \Prop$:]
\[
M, w \models p \Longleftrightarrow v_w (p) = \true~.
\]
\end{description}
The propositional connectives are treated classicaly.
\end{definition}

A Kripke model $M = \langle W, R, v\rangle $ is serial if for
every $w \in W$, $R[w] \neq \emptyset$.
$\D$ is the modal logic that is sound and
complete with
respect to serial Kripke models~\cite{halmos92}.

Now we can show the lower bound.

\begin{theorem}
\label{thm:lower_bound_PPL}
$\PPLSat$ is $\pspace$-hard.
\end{theorem}

\begin{proof}
We will reduce $\DSat$ to $\PPLSat$. Since $\DSat$ is
$\pspace$-complete~\cite{halmos92} our theorem follows.
Let $A \in \lanMod$ and let $f(A)$ be the
$\PPL$-formula obtained
from $A$ by replacing every occurrence
of $\Box$ by $P_{\geq 1}$. We will
show the following equivalence:
\begin{equation}
\label{eq:lower_bound_sat_eq}
A \text{ is $\D$-satisfiable if and only if } f(A)
\text{ is $\PPL$-satisfiable}~.
\end{equation}
($\Longrightarrow$)
Assume that $A$ is satisfiable. By
the small model theorem for modal logic 
$\D$~\cite{halmos92} it is satisfiable in a serial
Kripke model $M_D = \langle U, R,v\rangle$, where
$U$ is finite.
We define $M_{\PPL} = \langle U, W, H, \mu, 
v' \rangle $ where for every $w \in U$,
\begin{align*}
W_w & = R[w],\\
H_w & = \powerset (W_w),\\
v'_w & \text{ is an extension of $v_w$ to 
$\basis(\lanPPL)$-formulas},
\end{align*}
and for every $V \in H_w$,
\[
\mu_w (V) = \frac{|V|}{|R[w]|}~. 
\]
It is easy to show that $M_{\PPL}$ is
a $\PPL$-model. We will now show that
\[
(\forall w \in U)
(\forall B \in \subf(A))
\big [ M_D, w \models
B \Longleftrightarrow M_{\PPL}, w \models
f(B) \big ]
\]
by induction on the complexity of $B$. The
only interesting case is when $B$ is
of the form $\Box C$. Then
we have
\begin{align}
\nonumber
M_D , w & \models \Box C & 
\Longleftrightarrow\\
\nonumber
(\forall u \in R[w])
\big [ M_D , u & \models C \big ] & 
\stackrel{\ih}{\Longleftrightarrow}\\
\label{eq:lower_bound_1}
(\forall u \in W_w)
\big [ M_{\PPL} , u & \models f(C) \big ].
\end{align}
Now we have that $W_w = 
[f(C)]_{M_{\PPL},w}$ which
immediately implies that
\[
\mu_w ([f(C)]_{M_{\PPL},w}) = 1~.
\]
On the other hand assume that
$\mu_w ([f(C)]_{M_{\PPL},w}) = 1$.
Then if $[f(C)]_{M_{\PPL},w} \subsetneq
W_w$ then, by the definition of $\mu_w$
we get that $\mu_w (W_w \setminus [f(C)]_{M_{\PPL},w}) > 0$,
which, by the additivity of $\mu_w$
contradicts the fact that 
$\mu_w ([f(C)]_{M_{\PPL},w}) = 1$.
We conclude that
\eqref{eq:lower_bound_1} is equivalent
to the following:
\begin{align*}
\mu_w ([f(C)]_{M_{\PPL},w}) & = 1
& \Longleftrightarrow\\
M_{\PPL}, w \models & P_{\geq 1} f(C)
& \Longleftrightarrow\\
M_{\PPL}, w \models & f(B).
\end{align*}
($\Longleftarrow$) Assume that
$f(A)$ is satisfiable. Then $f(A)$ is
satisfiable in a model
$M_{\PPL} = \langle U, W, H, \mu, v\rangle$ that
satisfies the properties of \Fref{cor:smp_ppl}. 
Let $M_D = \langle U, R, v'\rangle$,
where, for every $w\in W$, 
\[
R[w] = \{ u \in W_w~|~ \mu_w(\{u\}) > 0\}
\]
and $v'_w$ is the restriction of $v_w$ to $\Prop$.
It is straightforward
to show that $M_D$ is a serial Kripke
structure. We
will now show that
\[
(\forall w \in W)
(\forall B \in \subf(A))
\big [ M_{\PPL}, w \models
f(B) \Longleftrightarrow M_D, w \models 
B \big ]
\]
by induction on the complexity of $B$.
The only interesting case is when
$B \equiv \Box C$. We have that
\begin{align}
\nonumber
M_D , w & \models \Box C & 
\Longleftrightarrow\\
\nonumber
(\forall u \in R[w])
\big [ M_D , u & \models C \big ] & 
\stackrel{\ih}{\Longleftrightarrow}\\
\label{eq:lower_bound_2}
(\forall u \in W_w)
\big [ \mu_w(\{u\}) > 0 & \Longrightarrow M_{\PPL} , u \models f(C) \big ]~.
\end{align}
So $[f(C)]_{M_{\PPL},w} \supseteq 
\{ u \in W_w ~|~ \mu_w (\{ u\} > 0)\}$.
Hence
\begin{align*}
\mu_w ([f(C)]_{M_{\PPL},w}) & =
\sum_{u \in [f(C)]_{M_{\PPL},w} }
\mu(\{ u\}) \\ 
& = \sum_{u \in W_w ~|~ \mu_w(\{ u \}) > 0 }
\mu(\{ u\}) = 
\sum_{u \in W_w}
\mu(\{ u\}) = 1 ~.
\end{align*}
On the other hand assume that
$\mu_w ([f(C)]_{M_{\PPL},w}) = 1$.
Let $u \in W_w$ such that 
$\mu_w(\{ u \}) > 0$. Assume that
$M_{\PPL},u \not\models f(C)$.
Then $u \in W_w \setminus 
[f(C)]_{M_{\PPL},w}$.
So $\mu_w (W_w \setminus 
[f(C)]_{M_{\PPL},w}) > 0$,
which contradicts the fact that
$\mu_w ([f(C)]_{M_{\PPL},w}) = 1$.
So for all $u \in W_w$, $\mu_w(\{ u\})>0$
implies that 
$M_{\PPL},u \models f(C)$.
We conclude that \eqref{eq:lower_bound_2}
is equivalent to
\begin{align*}
\mu_w ([f(C)]_{M_{\PPL},w}) & = 1
& \Longleftrightarrow\\
M_{\PPL},w \models
& P_{\geq 1} f(C)
& \Longleftrightarrow\\
M_{\PPL},w \models & f(B).
\end{align*}
We conclude that \eqref{eq:lower_bound_sat_eq}
holds, which proves our theorem. 
\end{proof}

\section{Applications}

In this section we apply the results of Sections
\ref{sec:non-iter} and \ref{sec:iter} in
probabilistic logics over classical propositional logic
and justification logic.

\label{sec:appl}

\subsection{Probabilistic Logics over Classical Propositional Logic}

Let $\CP$ denote classical propositional logic. If we define as basic formulas the atomic
propositions (i.e. elements of the set  $\Prop$) and if we define as evaluations
traditional truth assignments for classical propositional logic (on the set $\Prop$), then
we can define the non-iterated and the iterated probabilistic logic over classical
propositional logic, which according to our notation are $\PCP$ and $\PPCP$ respectively.
These logics have already been defined in \cite{ograma09} as $\LPPTwo$
and $\LPPOne$ respectively. 
So, we have have the following corollary:
\begin{corollary}
\label{cor:LLP1_LPP2_compl}
\begin{enumerate}
\item
\label{enum:LLP2_np_c}
$\LPPTwoSat$ is $\np$-complete.
\item
\label{enum:LLP1_pspace_c}
$\LPPOneSat$ is $\pspace$-complete.
\end{enumerate}
\end{corollary}
\begin{proof}
The satisfiability problem for
$\cpnb$-formulas in classical propositional logic can
be decided in polynomial time (we simply have to check 
whether the $\cpnb$ formula contains
an atomic proposition and its negation). So if we set
$\cclass=\p$ in Theorems \ref{thm:upper_bound_pl} and
\ref{thm:upper_bound_ppl}
we conclude that $\LPPTwoSat \in \np$ and that
$\LPPOneSat \in \pspace$.
The lower bounds follow from 
the fact that $\LPPTwo$ is an extension of classical propositional logic
and from \Fref{thm:lower_bound_PPL}.
\end{proof}

As it was observed in \cite{ograma09} the result
of  \Fref{cor:LLP1_LPP2_compl}\eqref{enum:LLP2_np_c} can be obtained
a straightforward application of the methods of \cite{fahame90}. Also in \cite{faginH94}, a proof
sketch (without many details) for
 \Fref{cor:LLP1_LPP2_compl}\eqref{enum:LLP1_pspace_c} was given. In this paper
we gave formal proofs for both results.

\subsection{Probabilistic Logics over Justification Logic}

Before defining the non-iterated and the iterated probabilistic 
justification logic, we briefly recall
justification logic $\J$ and its satisfiability
algorithm. 

\subsubsection{The basic Justification Logic \texorpdfstring{$\J$}{}}

The language of justification logic~\cite{artemovF16}
is defined by extending the language of classical
propositional logic with formulas of the form
$t : \alpha$ where $t$ is a justification
term, which is used to represent evidence, 
and $\alpha$ is a justification formula,
which is used to represent propositions,
statements or facts. As we will see later,
formula $\alpha$ might contain terms as well.
The formula
$t: \alpha$ reads as ''\emph{$t$ is a justification 
for believing $\alpha$}'' or as ''\emph{$t$ justifies
$\alpha$}''.
For example, assume that we have an agent who sees a snake behind
him/her. Whereas in traditional modal logic we can express a statement like ``the agent believes/knows that he/she is in danger'', in justification logic we can express a statement like ``the agent is in danger because there is a snake
behind him/her''. In the last statement an observation of the snake can serve as a justification. So, in justification
logic the representation of knowledge becomes explicit.

\emph{Justification terms} are built from countably many 
constants and countably many variables according to the 
following grammar:
\[
t :: = c ~|~ x ~|~ (t \cdot t) ~|~ (t+t) ~|~!t~,
\]
where $c$ is a constant and $x$ is a variable. 
$\Tm$ denotes the set of all terms
and $\Con$ denotes the set of all constants.
For $c \in \Con$ and $n \in \mathbb{N}$ we define
\[
!^0 c := c \qquad\text{and}\qquad
!^{n+1}c := {!} ~ ({!^n} c)~.
\]
The operators $\cdot$ and $+$ are assumed to be
left-associative.  The intended meaning of the connectives
used in the set $\Tm$ will be clear when we present the deductive system for
$\J$.
 
Formulas of the language $\lanJ$ (justification
formulas) are 
built according to the following grammar:
\[
\alpha :: = p~|~ t: \alpha ~|~ \lnot \alpha ~|~ \alpha \land \alpha~,
\]
where $t \in \Tm$ and $p \in \Prop$. Following our previous notation we have
\[
\basis(\lanJ) = \Prop \cup \{ t : \alpha ~|~ \alpha \in \lanJ \}~.
\]

The deductive system for $\J$ is the Hilbert 
system presented in \Fref{tab:system_J}.
Axiom $\axJ$ is 
also called the \emph{application axiom} and is the 
justification logic analogue of application axiom
in modal logic.
It states that we can combine a justification
for $\alpha \to \beta$ and a justification for
$\alpha$ in order to obtain a justification for $\beta$.
Axiom $\axPlus$, which is also called the
\emph{monotonicity axiom}, states that
if $s$ or $t$ is a justification for $\alpha$
then the term $s + t$ is also a justification
for $\alpha$. This operator can model monotone
reasoning like proofs in some formal system of
mathematics: if I already have a proof $t$ for a formula $\alpha$,
then $t$ remains a proof for $\alpha$ if I add a few more lines in $t$.
Rule $\ANE$ states that any constant
can be used to justify any axiom and also that we can use the
operator $!$ to express positive introspection:
if $c$ justifies axiom instance $\alpha$, then $!c$
justifies $c:\alpha$, $!!c$ justifies $!c:c:\alpha$ and so on. 
The previous situation is the explicit analogue of
the positive iteration of modalities in traditional
modal logic: I know
$\alpha$, I know that I know $\alpha$ and so on.
The operator $!$ is also called proof checker or proof verifier.
This is because we can think that $\alpha$ is a problem given
to a student, $c$ is the solution (or the proof) given by the student
and $!c$ is the verification of correctness for the proof
given by the tutor. So justification logic can model
the following situation:
\begin{description}
\item[\textbf{student:}] I have a proof for $\alpha$ (i.e. $c: \alpha$).
\item[\textbf{tutor:}]  I can verify your proof for $\alpha$ (i.e. $!c : c: \alpha$)
\end{description}

In justification logic it is common
to assume that only some constants justify some
axioms (see the notion of \emph{constant specification}
in \cite{artemovF16}). However, 
for the purposes of this paper
it suffices to assume that every constant justifies every
axiom (this assumption corresponds to the notion
of a \emph{total constant specification}~\cite{artemovF16}).

\begin{table}
{
\centering
\renewcommand*{\arraystretch}{1.25}
\begin{tabular}{|c l|}
\hline
Axioms:& \\
$\axP$ & finite set of axiom schemata axiomatizing
classical\\
& propositional logic in the language $\lanJ$\\
$\axJ$ & $\vdash s : (\alpha \to \beta) \to
( t :\alpha \to s \cdot t : \beta ) $\\
$\axPlus$ & $\vdash  ( s: \alpha \lor t : \alpha 
 ) \to  s+t: \alpha$\\
Rules: & \\
$\MP$ & \text{if }$T \vdash \alpha \text{ and } T 
\vdash \alpha \to \beta \text{ then } T \vdash 
\beta$\\
$\ANE$ & $\vdash {!^{n}} c : {!^{n-1}} c : \cdots : 
{!c} : c : \alpha$,
where $c$ is a constant, $\alpha$ is\\
& an instance of $\axP$, $\axJ$ or $\axPlus$ and 
$n \in \Nat$\\
\hline
\end{tabular}
\caption{\label{tab:system_J} The Deductive System $\J$}
}
\end{table}

In order to illustrate the usage of axioms and rules in $\J$ we present the following example:
\begin{example}
Let $a, b \in \Con$, $\alpha, \beta \in \lanJ$ and $x, y$ be variables. Then we have the following:
\[
\vdash_{\J} (x: \alpha \lor y:\beta) \to a \cdot x + b \cdot y: (\alpha \lor \beta).
\]
\end{example}

\begin{proof}
Since $\alpha \to \alpha \lor \beta$ and $\beta \to \alpha \lor \beta$ are instances of $\axP$, we can use
$\ANE$ to obtain
\[
\vdash_{\J} a : (\alpha \to \alpha \lor \beta)
\]
and
\[
\vdash_{\J} b : (\beta \to \alpha \lor \beta).
\]
Using  $\axJ$ and $\MP$ we obtain
\[
\vdash_{\J} x: \alpha \to a 
\cdot x : (\alpha \lor \beta)
\]
and
\[
\vdash_{\J} y: \beta \to b \cdot y : (\alpha \lor \beta).
\]

Using $\axPlus$ and propositional reasoning we obtain
\[
\vdash_{\J} x: \alpha \to a \cdot x +  b \cdot y : (\alpha \lor \beta)
\]
and
\[
\vdash_{\J} y: \beta \to a \cdot x +  b \cdot y : (\alpha \lor \beta).
\]
We can now obtain the desired result by applying propositional reasoning. 
\end{proof}

Logic $\J$ also enjoys the \emph{internalization property}, which is presented in the following theorem.
Internalization states that the logic internalizes its own notion of proof. The version without premises is an explicit form of the necessitation rule of modal logic. A proof of the following theorem can be found in~\cite{KuzStu12AiML}.

\begin{theorem}[Internalization]
For any $\alpha, ~\beta_1, \ldots , 
\beta_n \in \lanJ$ and $t_1, \ldots ,  t_n \in \Tm$,  if
\[
\beta_1, ~\ldots~,~ \beta_n \vdash_{\J} \alpha
\]
then there exists a term $t$ such that
\[
t_1 : \beta_1, ~\ldots~,~ t_n : \beta_n \vdash_{\J} t  : \alpha~.
\]
\end{theorem}

The models for $\J$ which
we are going to use in this paper are called $\mkr$-models
and were introduced by 
Mkrtychev~\cite{Mkr97LFCS}
for the logic $\LPlogic$. Later 
Kuznets~\cite{Kuz00CSLnonote} adapted these 
models for other justification logics
(including $\J$) and proved the corresponding
soundness and completeness theorems.
Formally, we have the following:

\begin{definition}[$\mkr$-Model]
\label{def:mkr_model}
An $\mkr$-model is a pair $\langle v, \evid
\rangle$, where $v: \Prop \to \{\true, \false\}$
and $\evid: \Tm \to \powerset(\lanJ)$
such that for every
$s, t \in \Tm$, for $c \in \Con$ and 
$\alpha,\beta \in \lanJ$, for $\gamma$ being an axiom
instance of $\J$ and $n \in \Nat$ we have
\begin{enumerate}
\item
$\big( \alpha \to \beta \in \evid (s) \text{ and }
\alpha \in \evid (t) \big)
\Longrightarrow \beta \in \evid (s \cdot t)$~;
\item
$\evid(s) \cup \evid(t) \subseteq \evid(s + t)$~;
\item
${!^{n-1} c} : {!^{n-2}} c: \cdots : !c : c : \gamma 
\in \evid (!^n c)$.
\end{enumerate}
\end{definition}
\begin{definition}[Truth in an $\mkr$-model]
We define what it means for an $\lanJ$-formula
to hold in the $\mkr$-model 
$M= \langle v, \evid \rangle$ inductively as follows (the
connectives $\lnot$ and $\land$ are treated classically):
\begin{align*}
M \models p & \Longleftrightarrow v(p) = \true \qquad
\text{for } p \in \Prop~;\\
M \models t : \alpha & \Longleftrightarrow \alpha \in \evid(t)~.
\end{align*}
\end{definition}

We close this section by briefly recalling 
the known complexity bounds for $\JSat$.
The next theorem
is due to Kuznets~\cite{Kuz00CSLnonote,Kuz08PhD}. 
We present it here
briefly using our own notation.

\begin{theorem}
\label{thm:upper_bound_cpnb_j}
The satisfiability problem for $\cpnb$-formulas in logic
$\J$ belongs to $\conp$.
\end{theorem}

\begin{proof}
Let $a$ be the $\cpnb$-formula of logic $\J$ that is
tested for satisfiability.
Assume that there is no $p \in \Prop$
such that $p$ appears both
positively and negatively in $a$ (otherwise it is clear
that $a$ is not satisfiable). So, the
satisfiability of $a$ depends only on the justification
assertions that appear in $a$.

Let $p_i:\alpha_i$ be the assertions that appear
positively in $a$ and let $n_i : \beta_i$
be the assertions that appear negatively in $a$.
A short
no-certificate is an object that has size polynomial in the
size of the input (i.e. of $|a|$) and that can witness that the
input is a no-instance of the problem in polynomial time.
So, if we can show  that the question
``does a model, which satisfies all the
$p_i : \alpha_i$'s and falsifies all the $n_i : \beta_i$'s,
exist?'' has a short no-certificate we have proved claim of the theorem.
For this purpose it suffices to guess some 
$n_j : \beta_j$ (i.e. some justification assertion that appears negatively in
$a$) and show
that every $\mkr$-model that satisfies all the 
$p_i : \alpha_i$'s, satisfies $n_j : \beta_j$ too.
Our guess will not only consist of the formula itself, but
also of the way this formula is constructed
from the $p_i : \alpha_i$'s and from the constants
that justify the axioms of $\J$. In order to verify that
our guess of the formula and its construction
is correct we have to be able to access all the formulas that are justified by a given term in finite time.
At first, it seems impossible to do this in finite time, 
since we have that some terms justify infinitely many
formulas (in particular 
every constant justifies all the axiom instances, which are 
infinitely many). However, since $\J$ is axiomatized
by finitely many axiom schemes we can use schematic
variables for formulas and terms. This way we have that
every constant justifies only finitely many axiom
schemes. So, the short no-certificate can be guessed 
as follows:
we non-deterministically choose some $n_j:\beta_j$.
The term $n_j$ is created by the connectives $\cdot$
and $+$ using a finite tree like the following one:
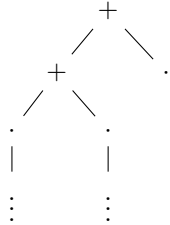
\begin{center}
\begin{tikzpicture}
\matrix (m) [matrix of math nodes, row sep = 10pt, column sep = 5pt]
{
& & +\\
& + & & & \cdot \\
\cdot & & \cdot \\
\vdots & & \vdots \\
};
\path (m-1-3) edge (m-2-2);
\path (m-1-3) edge (m-2-5);
\path (m-2-2) edge (m-3-1);
\path (m-2-2) edge (m-3-3);
\path (m-3-1) edge (m-4-1);
\path (m-3-3) edge (m-4-3);
\end{tikzpicture}
\end{center}
The above tree can be constructed in many ways
but we guess one such that
in the leaves there are terms of
the form $!^{n} c: \cdots : !c : c$ or some
$p_i$'s. Such a guess should be possible if the
no-certificate exists. 
The fact that the tree is constructed in that way means, $n_j : \beta_j$
is constructed by the axioms $\axPlus$, $\axJ$, the rule $\ANE$,
using the assumptions that every $p_i$ satisfies every $\alpha_i$.
We assign to each of the tree's leaves an
axiom scheme (if the leaf is of the form
$!^{n} c: \cdots : !c : c$) or an $\alpha_i$ 
(if the leaf is of the form $p_i$). Then starting from the leaves (i.e. bottom-up)
we unify the schematic formulas assigned to every tree node. If at some point
the unification is impossible, then our guess is not correct.
The bottom-up unification procedure ends when we construct a formula in the root.
This formula has to be unifiable with $\beta_j$, otherwise the guess is again not correct.
The formulas that we assigned to the leaves and
the structure of the tree compose the short
no-certificate. It can be shown that representing
formulas as directed acyclic graphs and using
Robinson's unification algorithm~\cite{corbinB83} we 
can verify that the
unifications succeed in polynomial time. Hence the
satisfiability for $\cpnb$-formulas in logic $\J$
belongs to $\conp$. 
\end{proof}

Finally, we can present the known complexity
bounds for $\JSat$.

\begin{theorem}
\label{thm:J_compl_bounds}
$\JSat$ is $\sig{p}{2}$-complete.
\end{theorem}

\begin{proof}
Let $\alpha \in \lanJ$ be the formula that is tested
for satisfiability. The upper bound follows by Kuznets's algorithm~\cite{Kuz00CSLnonote} which can be
described by the following steps, using our notation:
\begin{enumerate}
\item 
Create a node with $w~\true~\alpha$. Apply the propositional rules for as long as possible.
Non-deterministically choose a world path $p_{w_{ij}}$. 
This can be done in non-determinstic polynomial time.
\item
Verify that $B_{p_{w_{ij}}}$ is satisfiable using the $\conp$-algorithm of \Fref{thm:upper_bound_cpnb_j}.
\end{enumerate}
Observe that there is no rule that produces new worlds
in the above algorithm. We have presented
the algorithm using world signs in order to be
consistent with our tableau notation.
The lower bound follows by a result of 
Achilleos~\cite{ach15}. 
\end{proof}

\subsubsection{Probabilistic Justification Logic}

Let $M$ be an $\mkr$-model. Based on $M$ we can define the 
 evaluation $v_M$ as follows:
\begin{center}
for every $\beta \in \basis(\lanJ)$, $v_M(\beta) = \true$
if and only if $M \models \beta$.
\end{center}

So, if we set $\logic = \J$, $\lan=\lanJ$ and we define the
evaluations as above, 
we can define the non-iterated probabilistic
logic over $\J$, which is the logic $\PJ$.
For the complexity of the satisfiability problem in $\PJ$, we have the following corollary:

\begin{corollary}
$\PJSat$ is $\sig{p}{2}$-complete.
\end{corollary}

\begin{proof}
The upper bound follows from 
Theorems \ref{thm:upper_bound_cpnb_j} and 
\ref{thm:upper_bound_pl}.
Since $PJ$ is an extension of $\J$ and $\JSat$ is $\sig{p}{2}$-hard (\Fref{thm:J_compl_bounds})
we get the lower bound. 
\end{proof}

Now we will present the iterated probabilistic justification logic $\PPJ$. 
According to our previous definitions, in order to
define the language of $\PPJ$, which is called $\lanPPJ$, it suffices to define the basic formulas of $\PPJ$.
So, we define  $\basis(\lanPPJ)  = \Prop \cup \{ t : A ~|~ A \in \lanPPJ \}$.
For the convenience of the reader we give the complete definition of $\lanPPJ$, i.e.
\[
A ::= p ~|~ t:A ~|~ \lnot A~|~ A \land A ~|~ P_{\geq s} A,
\]
where $p \in \Prop$ and $s \in \mathbb{Q} \cap [0,1]$. We observe that
the basic formulas of $\PPJ$ may contain random formulas of $\PPJ$.
So, the reader might think that the basic formulas do not have much simpler structure than
the normal formulas of $\PPJ$ . However, the complex formulas of $\PPJ$ can appear only in the
scope of justification terms. In $\PPJ$, since there is no axiom that can draw a formula
from the scope of a justification term (e.g. $t: A \to A$), formulas that
appear under the scope of a justification operator practically behave as atomic propositions.
So, decidability and complexity of basic $\PPJ$-formulas is much simpler to prove than
decidability and complexity of random $\PPJ$-formulas.

In order to define the $\PPJ$-models we have to define evaluations for basic $\PPJ$-formulas. Since
basic $\PPJ$-formulas resemble basic $\J$-formulas, it makes sense to extend the definition
of $\mkr$-models in order to define evaluations for $\PPJ$. So,
as in the case of logic $\J$ we have
to present the deductive system of $\PPJ$ first.
This system is presented in
\Fref{tab:systemPPJ}. The axiomatization of $\PPJ$
is a combination of the axiomatization for $\LPPOne$~\cite{ograma09} and of the axiomatization for
the basic justification logic $\J$.
Axiom $\pos$ corresponds to the fact that the
probability of truthfulness of every 
formula is at least $0$ (the
acronym $\pos$ stands for 
$\mathsf{n}$on-$\mathsf{n}$egative).
Observe that by substituting 
$\lnot A$ for $A$ in $\pos$, we have $P_{\geq 0} \lnot A$,
which by our syntactical abbreviations is
$P_{\leq 1} A$. Hence axiom $\pos$ also corresponds to
the fact that the probability of
truthfulness for every formula
is at most $1$.
Axioms $\LOne$ and $\LTwo$ describe some properties of
inequalities (the $\mathsf{L}$ in
$\LOne$ and $\LTwo$ stands for $\mathsf{l}$ess).
Axioms $\AddOne$ and $\AddTwo$ correspond to the additivity of probabilities for disjoint events (the $\mathsf{Add}$ in $\AddOne$ and $\AddTwo$
stands for $\mathsf{add}$itivity).
Rule $\CE$ is the probabilistic analogue
of the necessitation rule in
modal logics (hence the acronym $\CE$
stands for $\mathsf{p}$robabilist 
$\mathsf{n}$ecessitation): if a 
formula is valid, then it has probability $1$. 
Rule $\ST$ intuitively states that if the
probability of a formula is
arbitrary close to $s$, then it is at least $s$.
Observe that the rule $\ST$ is infinitary in the sense that
it has an infinite number of premises. 
It corresponds to the 
Archimedean property for the real numbers. The acronym $\ST$
stands for $\mathsf{st}$rengthening, since the 
statement of the
result is stronger than the statement of the premises.
Rule $\ST$ was introduced in~\cite{ognjanovicR00,ravskovicO99}
so that strong completeness
for probabilistic logics could be proved.
We recall that a logical system is strongly
complete if and only if every consistent set (finite or
infinite) has a model. As it is shown in \cite{ognjanovicR00,ravskovicO99},
languages used for probabilistic logics are
non-compact, so the proof of strong completeness is
impossible without an infinitary rule.

\begin{table}
{
\centering
\renewcommand*{\arraystretch}{1.25}
\begin{tabular}{|c l|}
\hline 
Axioms: &\\
$\axP$ & 
finitely many axiom schemata axiomatizing\\
& classical propositional logic in the language $\lanPPJ$\\
$\pos$ & $\vdash P_{\geq 0} A$\\
$\LOne$ & $\vdash P_{\leq r} A \to P_{< s} A$, where $s > r$\\
$\LTwo$ & $\vdash P_{< s} A \to P_{\leq s} A$\\
$\AddOne$ & $\vdash  P_{\geq r} A \land P_{\geq s} B \land P_{\geq 1} \lnot (A \land B) \to P_{\geq \min(1, r+s)} (A \lor B)$\\
$\AddTwo$ & $\vdash P_{\leq r} A \land P_{< s} B \to P_{<r+s} (A \lor B)$, where $r+s \leq 1$\\
$\axJ$ & $\vdash s : (A \to B) \to
( t :A \to s \cdot t : B ) $\\
$\axPlus$ & $\vdash  ( s: A \lor t : A  ) \to  s+t: A$\\
Rules: &\\
$\MP$ & if $T \vdash A$ and  $ T \vdash A \to B$ 
then $T \vdash B$\\
$\CE$ & if $\vdash A$ 
then $\vdash P_{\geq 1 } A$\\
$\ST$ & if $T \vdash A \to P_{\geq s - \frac{1}{k} } B$
for every integer $k \geq \frac{1}{s}$ and $s > 0$\\
& then $T \vdash A \to P_{\geq s } B$\\
$\ANE$ & $\vdash {!^{n}} c : {!^{n-1}} c : \cdots : 
{!c} : c : A$, where $c \in \Con$, $A$ is\\
& an instance of some $\PPJ$-axiom and $n \in \Nat$\\
\hline
\end{tabular}
\caption{\label{tab:systemPPJ} The Deductive System $\PPJ$}
}
\end{table}

So, now we can extend \Fref{def:mkr_model} to the basic formulas
of $\lanPPJ$. Recall that basic $\PPJ$-formulas may contain iterated
probabilistic operators under the scope of justification terms.
This is the reason why we cannot use the standard $M$-models for
basic $\PPJ$-formulas. However since, as we mentioned before,
$\PPJ$-formulas under the scope of justification operators
practically behave as atomic propositions the only difference
between between extended $\mkr$-models and normal ones
is that in the former we have the logic $\PPJ$ in the place
where the later have the logic $\J$.

\begin{definition}[Extended $\mkr$-Model]
An extended $\mkr$-model is a pair $\langle v, \evid
\rangle$, where $v: \Prop \to \{\true, \false\}$
and $\evid: \Tm \to \powerset(\lanPPJ)$
such that for every
$s, t \in \Tm$, for $c \in \Con$ and 
$A, B \in \lanPPJ$, for $C$ being an axiom
instance of $\PPJ$ and $n \in \Nat$ we have
\begin{enumerate}
\item
$\big(A \to B \in \evid (s) \text{ and } A \in \evid (t) \big)
\Longrightarrow B \in \evid (s \cdot t)$~;
\item
$\evid(s) \cup \evid(t) \subseteq \evid(s + t)$~;
\item
${!^{n-1} c} : {!^{n-2}} c: \cdots : !c : c : C \in \evid (!^n c)$.
\end{enumerate}
\end{definition}
We can define evaluations based on extended $\mkr$-models 
 in the same way as for the standard $\mkr$-models.
So a $\PPJ$-model is a $\PPL$-model where the evaluations are
based on extended $\mkr$-models. This completes the definition
of semantics for $\PPJ$.

Now we are ready to present the complexity bounds for $\PPJSat$.
\begin{theorem}
\label{thm:upper_bound_cpnb_PPL}
The satisfiability problem for $\cpnb$-formulas
in logic $\PPJ$ belongs to $\conp$.
\end{theorem}

\begin{proof}
We recall that $\cpnb$-formulas in $\lanPPJ$
are conjunctions of positive and negative atomic propositions
and positive and negative formulas of the form $t:A$, where
$A \in \lanPPJ$.

Since the $\cpnb$-formulas in 
$\lanPPJ$ resemble the
$\cpnb$-formulas in $\lanJ$ we will
use a slight variation of the algorithm in \Fref{thm:upper_bound_cpnb_j}
for deciding the satisfiability of $\cpnb$-formulas.
It is not difficult to observe that the algorithm of \Fref{thm:upper_bound_cpnb_j} does
not depend on what the axioms of the logic are, as long
as they are finitely many. Assume that $a$ is the $\cpnb$-formula
of logic $\PPJ$ that we want to test for satisfiability.
As in the proof of \Fref{thm:upper_bound_cpnb_j}
we can choose a formula $t:B$ that appears negatively in $a$, guess a tree that describes the
construction of $t$, assign to the leaves of the tree $\PPJ$-axioms and
formulas and then verify that our guess is correct using unification. 
However, since the axioms of $\PPJ$ have different form than
the axioms of $\J$ we have to modify the unification algorithm. 
Simple unification is not sufficient any more, since the axioms
of $\PPJ$ come with linear conditions. In the rest of the proof
we explain that the verification can be done in polynomial time by
using a unification algorithm and by testing a linear system for satisfiability.

Whereas for $\lanJ$ we need two kinds of schematic variables (for terms and formulas),
for $\lanPPJ$ we need three kinds of
schematic variables: for terms, formulas and
rational numbers. Also, because of the
side conditions that come with the axioms 
$\LOne$ and $\AddTwo$ our schematic formulas 
should be paired with systems of linear inequalities.
For example, the scheme $\LOne$ should be represented 
by the schematic formula
$P_{\leq r} A \to P_{< s} A$ (with the schematic 
variables $r$, $s$, and~$A$) together with the 
inequality $r <s$,
whereas
a scheme that is obtained by a conjunction
of the schemata $\LOne$ and $\AddTwo$ should be
represented as
\[
\big( P_{\leq r_1} A_1 \to P_{< s_1} A_1 \big ) \land 
\big ( P_{\leq r_2} A_2 \land P_{< s_2} B_2 \to 
P_{<r_2+s_2} (A_2 \lor B_2) \big )
\]
together with the inequalities
\[
\big \{ r_1 <s_1, r_2 + s_2 \leq 1 \big \}~.
\]
We should not forget that the rational variables
belong to $\Rat \cap [0,1]$. 
So we have to add constraints like
$0 \leq r \leq 1$.
Hence in addition to
constructing unification equations we need to
take care of the linear constraints.
For instance, in order to unify
the schemata $P_{\geq r} A$ and 
$P_{\geq s} B$ the algorithm has to unify $A$ and~$B$, 
and to equate $r$ and $s$, i.e.~it adds $r=s$ to the
linear system. At the end the verification algorithm
will succeed only if the standard unification of formulas
succeeds and the linear system is solvable. 

Another complication are constraints
of the form 
\begin{equation}
\label{eq:nonlin:1}
l = \min(1, r + s)
\end{equation}
that originate from the scheme $\AddOne$.
Obviously, Eq.~\eqref{eq:nonlin:1} is not linear.
However, we find that  Eq.~\eqref{eq:nonlin:1}
has a solution if and only if one of the set of equations
\[
\{l = r+s, r+s\leq 1 \}
\text{ or }
\{l = 1, r+s > 1\}
\]
has a solution. Thus whenever we come to
an equation like~\eqref{eq:nonlin:1} we can non-deterministically
chose one of the equivalent set of equations and add it to the constructed linear system.

We conclude that we can guess a tree for $t$ and also a linear system in non-deterministic polynomial time.
We also find that the verification can be done in polynomial
time, since testing a linear system for satisfiability
can be done
in polynomial time \cite{Karmarkar84} and unification of formulas
can be checked in polynomial time using Robinson's algorithm. So, as in the case of $\J$ we can
show that the satisfiability problem
for $\cpnb$-formulas in $\PPJ$ has short no-certificates, i.e it belongs
to $\conp$.
\end{proof}

\begin{remark}
It might seem strange to the reader that the $\cpnb$-formulas of $\J$ and $\PPJ$ have the
same upper complexity bound. At first sight the $\PPJ$ $\cpnb$-formulas seem more complex,
since they may contain probabilistic operators (and iterations of them). However in 
the $\PPJ$ $\cpnb$-formulas, 
the probabilistic operators, may occur only under the scope of justification terms. This means that
the probabilistic operators cannot impose any conditions for satisfiability. For example in $t : P_{\geq s} A$,
the formula $P_{\geq s} A$ does not have to be satisfiable. So, probabilistic operators under the
scope of justification operators, behave as atomic propositions. In particular they are more
sophisticated than simple atomic propositions because they consist of rational numbers and formulas.
It is not difficult to see, that the only impact that this has on the standard decidability procedure 
for the logic $\J$ is
on the unification algorithm. That is why the biggest part of
the proof for \Fref{thm:upper_bound_cpnb_PPL} is devoted on showing
how the modified unification algorithm works and that it runs in polynomial time.
\end{remark}

\begin{corollary}
$\PPJSat$ is $\pspace$-complete.
\end{corollary}

\begin{proof}
The upper bound follows from 
Theorems \ref{thm:upper_bound_cpnb_PPL} and
\ref{thm:upper_bound_ppl}.
The lower bound follows from \Fref{thm:lower_bound_PPL}.
\end{proof}

\section{Conclusion}
\label{sec:concl}

We have presented upper and lower complexity bounds for the satisfiability
problem in non-iterated and iterated probabilistic logics over any extension of classical
propositional logic. The aforementioned bounds are parameterized on the complexity of
satisfiability of conjunctions of positive
and negative formulas  that have neither a probabilistic nor a classical operator as their top-connectives. As an
application we have shown how tight bounds for the complexity of satisfiability in
non-iterated and iterated probabilistic logics
over classical propositional logic and justification
logic can be obtained.
It is interesting that both for classical
propositional logic and for the basic justification
logic $\J$ adding non-nested
probabilistic operators to the language does not increase the
complexity of the satisfiability problem.

Now we present some directions for further research.
The probabilistic logic of \cite{doderSO18_decidable} allows iterations of
some probabilistic operators, that are more complicated than the ones used in this
paper. It would be interesting to check whether the 
tableau procedure of our paper is applicable in this logic.
Another interesting question is the following:
Ra\v{s}kovi\'{c} et al.~\cite{raskovicMO08} define a 
probabilistic logic over classical propositional logic, 
called $\LPPS$, using
approximate conditional probabilities with the intention
to model non-monotonic reasoning. The satisfiability
problem for $\LPPS$ is again reduced to solving linear
systems. In~\cite{ognjanovicSS17} the logic of
\cite{raskovicMO08} is extended to a probabilistic
logic with approximate conditional probabilities over
justification logic. What are the complexity bounds for the
satisfiability problem
in the logics of \cite{raskovicMO08} and 
\cite{ognjanovicSS17}?

\bibliographystyle{abbrv}

\end{document}